\documentclass[11pt,a4paper]{article}

\usepackage{fullpage}
\usepackage[utf8x]{inputenc}
\usepackage{amsmath, amsthm}
\usepackage{wrapfig,graphicx,amssymb,textcomp,array,wasysym}
\usepackage{algpseudocode} 
\usepackage{enumerate}
\usepackage{multirow}
\algtext*{EndWhile}
\algtext*{EndIf}
\usepackage{algorithm}
\usepackage{color}

\newcommand{\MST}{\text{\em MST}}

\newcommand{\SMGG}{\text{\em Strong-matching}}

\newcommand{\dg}[1]{deg(#1)}

\newcommand{\nnb}[1]{\overline{N}(#1)}

\newcommand{\tr}[1]{t(#1)}
\newcommand{\tru}[1]{t'(#1)}
\newcommand{\trmin}{t}
\newcommand{\emin}{e}
\newcommand{\cmin}{D(u,v)}
\newcommand{\cone}[2]{C^{#1}_{#2}}
\newcommand{\hex}[2]{X(#1,#2)}
\newcommand{\tra}[2]{{#1}({#2})}
\newcommand{\G}[2]{G_{#1}({#2})}
\newcommand{\Inf}[1]{\text{Inf}(#1)}
\newcommand{\disc}{\Circle}
\newcommand{\discs}{\scriptsize \Circle}
\newcommand{\ddisc}{\ominus}
\newcommand{\ddiscs}{\scriptsize \ominus}
\newcommand{\sqr}{\Box}
\newcommand{\sqrs}{\scriptsize\Box}
\newcommand{\trid}{\bigtriangledown}
\newcommand{\trids}{\scriptsize\bigtriangledown}
\newcommand{\triu}{\bigtriangleup}
\newcommand{\trius}{\scriptsize\bigtriangleup}

\newcommand{\GD}{G_{\bigtriangledown}}
\newcommand{\GU}{G_{\bigtriangleup}}
\newcommand{\GUD}{G_{\bigtriangledown \hspace*{-7.5pt} \bigtriangleup}}

\newcommand{\sq}{\square}
\newcommand{\GS}[1]{G_{\sq}(#1)}
\newcommand{\SP}[2]{S^{\text{\tiny +#1}}_\text{\tiny #2}}
\newcommand{\SM}[2]{S^{\text{-\tiny #1}}_\text{\tiny #2}}

\title{Strong Matching of Points with Geometric Shapes\thanks{Research supported by NSERC.}}

\author{
Ahmad Biniaz\thanks{School of Computer Science, Carleton University, Ottawa, Canada.}
\and 
Anil Maheshwari\footnotemark[2]
\and 
Michiel Smid\footnotemark[2]
}
\date{\today}
\newtheorem{lemma}{Lemma}
\newtheorem{corollary}{Corollary}
\newtheorem{conjecture}{Conjecture}
\newtheorem{theorem}{Theorem}
\newtheorem{observation}{Observation}
\newtheorem{definition}{Definition}
\begin{document}

\maketitle

\begin{abstract}
Let $P$ be a set of $n$ points in general position in the plane. Given a convex geometric shape $S$, a geometric graph $\G{S}{P}$ on $P$ is defined to have an edge between two points if and only if there exists an empty homothet of $S$ having the two points on its boundary. A matching in $\G{S}{P}$ is said to be {\em strong}, if the homothests of $S$ representing the edges of the matching, are pairwise disjoint, i.e., do not share any point in the plane. We consider the problem of computing a strong matching in $\G{S}{P}$, where $S$ is a diametral-disk, an equilateral-triangle, or a square. We present an algorithm which computes a strong matching in $\G{S}{P}$; if $S$ is a diametral-disk, then it computes a strong matching of size at least $\lceil \frac{n-1}{17} \rceil$, and if $S$ is an equilateral-triangle, then it computes a strong matching of size at least $\lceil \frac{n-1}{9} \rceil$. If $S$ can be a downward or an upward equilateral-triangle, we compute a strong matching of size at least $\lceil \frac{n-1}{4} \rceil$ in $\G{S}{P}$. When $S$ is an axis-aligned square we compute a strong matching of size $\lceil \frac{n-1}{4} \rceil$ in $\G{S}{P}$, which improves the previous lower bound of $\lceil \frac{n}{5} \rceil$. 
\end{abstract}

\section{Introduction}
\label{intro}

Let $S$ be a compact and convex set in the plane that contains the origin in its interior. A {\em homothet} of $S$ is obtained by scaling $S$ with respect to the origin by some factor $\mu\ge 0$, followed by a translation to a point $b$ in the plane: $b+\mu S=\{b+\mu a: a\in S\}$.
For a point set $P$ in the plane, we define $\G{S}{P}$ as the geometric graph on $P$ which has an straight-line edge between two points $p$ and $q$ if and only if there exists a homothet of $S$ having $p$ and $q$ on its boundary and whose interior does not contain any point of $P$. If $P$ is in ``general position'', i.e., no four points of $P$ lie on the boundary of any homothet of $S$, then $\G{S}{P}$ is plane~\cite{Bose2010}. Hereafter, we assume that $P$ is a set of $n$ points in the plane, which is in general position with respect to $S$ (we will define the general position in Section~\ref{preliminaries}). 
If $S$ is a disk $\disc$ whose center is the origin, then $\G{\discs}{P}$ is the Delaunay triangulation of $P$. If $S$ is an equilateral triangle $\trid$ whose barycenter is the origin, then $\G{\trids}{P}$ is the triangular-distance Delaunay graph of $P$ which is introduced by Chew~\cite{Chew1989}.

A {\em matching} in a graph $G$ is a set of edges which do not share any vertices. A {\em maximum matching} is a matching with maximum cardinality. A {\em perfect matching} is a matching which matches all the vertices of $G$. Let $\mathcal{M}$ be a matching in $\G{S}{P}$. $\mathcal{M}$ is referred to as a {\em matching of points with shape $S$}, e.g., a matching in $\G{\discs}{P}$ is a matching of points with with disks. Let $\mathcal{S}_{\mathcal{M}}$ be a set of homothets of $S$ representing the edges of $\mathcal{M}$. $\mathcal{M}$ is called a {\em strong matching} if there exists a set $\mathcal{S}_{\mathcal{M}}$ whose elements are pairwise disjoint, i.e., the objects in $\mathcal{S}_{\mathcal{M}}$ do not share any point in the plane. Otherwise, $\mathcal{M}$ is a {\em weak matching}. See Figure~\ref{strong-example}. To be consistent with the definition of the matching in the graph theory, we use the term ``matching'' to refer to a weak matching. Given a point set $P$ in the plane and a shape $S$, the {\em (strong) matching problem} is to compute a (strong) matching of maximum cardinality in $\G{S}{P}$.
In this paper we consider the strong matching problem of points in general position in the plane with respect to a given shape $S\in \{\ddisc,\trid, \sqr\}$ (see Section~\ref{preliminaries} for the definition), where by $\ddisc$ we mean the line segment between the two points on the boundary of the disk is a diameter of that disk.

\begin{figure}[htb]
  \centering
\setlength{\tabcolsep}{0in}
  $\begin{tabular}{ccc}
\multicolumn{1}{m{.33\columnwidth}}{\centering\includegraphics[width=.22\columnwidth]{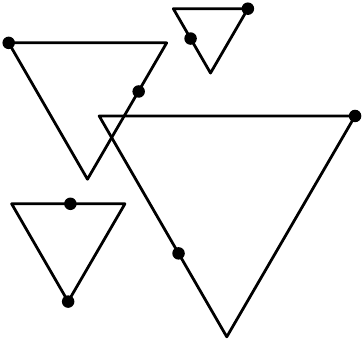}}
&\multicolumn{1}{m{.33\columnwidth}}{\centering\includegraphics[width=.22\columnwidth]{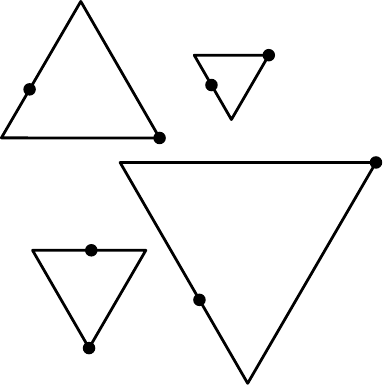}} &\multicolumn{1}{m{.33\columnwidth}}{\centering\includegraphics[width=.22\columnwidth]{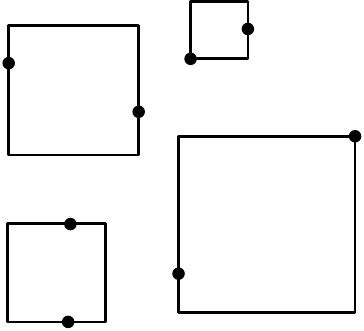}}
\\
(a) & (b)& (c)
\end{tabular}$
  \caption{Point set $P$ and (a) a perfect weak matching in $\G{\trids}{P}$, (b) a perfect strong matching in $\GUD(P)$, and (c) a perfect strong matching in $
\G{\sqr}{P}$.}
\label{strong-example}
\end{figure}

\subsection{Previous Work}
\label{previous-work}
The problem of computing a maximum matching in $\G{S}{P}$ is one of the fundamental problems in computational geometry and graph theory \cite{Abrego2004, Abrego2009, Babu2014, Bereg2009, Biniaz2014, Biniaz2015, Dillencourt1990}. 
Dillencourt~\cite{Dillencourt1990} and \'{A}brego et al. \cite{Abrego2004} considered the problem of matching points with disks. Let $S$ be a closed disk $\disc$ whose center is the origin, and let $P$ be a set of $n$ points in the plane which is in general position with respect to $\disc$. Then, $\G{\discs}{P}$ is the graph which has an edge between two points $p,q\in P$ if there exists a homothet of $\disc$ having $p$ and $q$ on its boundary and does not contain any point of $P\setminus\{p,q\}$. $\G{\discs}{P}$ is equal to the Delaunay triangulation on $P$, $DT(P)$. Dillencourt~\cite{Dillencourt1990} proved that $\G{\discs}{P}$ contains a perfect (weak) matching. \'{A}brego et al. \cite{Abrego2004} proved that $\G{\discs}{P}$ has a strong matching of size at least $\lceil(n-1)/8\rceil$. They also showed that there exists a set $P$ of $n$ points in the plane with arbitrarily large $n$, such that $\G{\discs}{P}$ does not contain a strong matching of size more than $\frac{36}{73}n$.

For two points $p$ and $q$, the disk which has the line segment $pq$ as its diameter is called the diametral-disk between $p$ and $q$. We denote a diametral-disk by $\ddisc$. Let $\G{\ddiscs}{P}$ be the graph which has an edge between two points $p,q\in P$ if the diametral-disk between $p$ and $q$ does not contain any point of $P\setminus\{p,q\}$. $\G{\ddiscs}{P}$ is equal to the Gabriel graph on $P$, $GG(P)$. Biniaz et al.~\cite{Biniaz2014} proved that $\G{\ddiscs}{P}$ has a matching of size at least $\lceil(n-1)/4\rceil$, and this bound is tight.

The problem of matching of points with equilateral triangles has been considered by Babu et al.~\cite{Babu2014}.
Let $S$ be a downward equilateral triangle $\trid$ whose barycenter is the origin and one of its vertices is on the negative $y$-axis. Let $P$ be a set of $n$ points in the plane which is in general position with respect to $\trid$. Let $\G{\trids}{P}$ be the graph which has an edge between two points $p,q\in P$ if there exists a homothet of $\trid$ having $p$ and $q$ on its boundary and does not contain any point of $P\setminus\{p,q\}$. $\G{\trids}{P}$ is equal to the triangular-distance Delaunay graph on $P$, which was introduced by Chew~\cite{Chew1989}. Bonichon et al.~\cite{Bonichon2010} showed that $\G{\trids}{P}$ is equal to the half-theta six graph on $P$, $\frac{1}{2}\Theta_6(P)$. Babu et al.~\cite{Babu2014} proved that $\G{\trids}{P}$ has a matching of size at least $\lceil(n-1)/3\rceil$, and this bound is tight. If we consider an upward triangle $\triu$, then $\G{\trius}{P}$ is defined similarly. Let $\GUD(P)$ be the graph on $P$ which is the union of $\G{\trids}{P}$ and $\G{\trius}{P}$. Bonichon et al.~\cite{Bonichon2010} showed that $\GUD(P)$ is equal to the theta six graph on $P$, $\Theta_6(P)$. Since $\G{\trids}{P}$ is a subgraph of $\GUD(P)$, the lower bound of $\lceil(n-1)/3\rceil$ on the size of maximum matching in $\G{\trids}{P}$ holds for $\GUD(P)$.  

The problem of strong matching of points with axis-aligned rectangles is trivial. An obvious algorithm is to repeatedly match the two leftmost points. The problem of matching points with axis-aligned squares was considered by \'{A}brego et al. \cite{Abrego2009}.
Let $S$ be an axis-aligned square $\sqr$ whose center is the origin. Let $P$ be a set of $n$ points in the plane which is in general position with respect to $\sqr$. Let $\G{\sqrs}{P}$ be the graph which has an edge between two points $p,q\in P$ if there exists a homothet of $\sqr$ having $p$ and $q$ on its boundary and does not contain any point of $P\setminus\{p,q\}$. $\G{\sqrs}{P}$ is equal to the $L_\infty$-Delaunay graph on $P$. \'{A}brego et al. \cite{Abrego2004, Abrego2009} proved that $\G{\sqrs}{P}$ has a perfect (weak) matching and a strong matching of size at least $\lceil n/5\rceil$. Further, they showed that there exists a set $P$ of $n$ points in the plane with arbitrarily large $n$, such that $\G{\sqrs}{P}$ does not contain a strong matching of size more than $\frac{5}{11}n$. Table~\ref{table1} summarizes the results.

Bereg et al.~\cite{Bereg2009} concentrated on matching points of $P$ with axis-aligned rectangles and squares, where $P$ is not necessarily in general position. 
They proved that any set of $n$ points in the plane has a strong rectangle matching of size at least $\lfloor\frac{n}{3}\rfloor$, and such a matching can be computed in $O(n \log n)$ time. As for squares, they presented a $\Theta(n\log n)$ time algorithm that decides whether a given matching has a weak square realization, 
and an $O(n^2\log n)$ time algorithm for the strong square matching realization. They also proved that it is NP-hard to decide whether a given point set has a perfect strong square-matching. 

\begin{table}
\centering
\begin{minipage}{13cm}\centering
\caption{Lower bounds on the size of weak and strong matchings in $\G{S}{P}$.}
\label{table1}
    \begin{tabular}{|c|c|c||c|c|}
         \hline
             $S$ 	& weak matching & reference& strong matching &reference \\  \hline  \hline
	      {$\disc$}& 	${\lfloor \frac{n}{2}\rfloor}$&\cite{Dillencourt1990}& $\lceil\frac{n-1}{8}\rceil$&\cite{Abrego2004}\\
	      {$\ominus$} &${\lceil \frac{n-1}{4}\rceil}$&\cite{Biniaz2014} &$\lceil\frac{n-1}{17}\rceil$&Theorem~\ref{Gabriel-thr}\\  
             {$\bigtriangledown$} &${\lceil \frac{n-1}{3}\rceil}$&\cite{Babu2014} &$\lceil\frac{n-1}{9}\rceil$&Theorem~\ref{half-theta-six-thr}\\  
	    {$\trid$ or $\triu$} &$\lceil \frac{n-1}{3}\rceil$& \cite{Babu2014} &$\lceil\frac{n-1}{4}\rceil$&Theorem~\ref{theta-six-thr}\\ \hline\hline
	    \multirow{2}{*}{$\Square$}& \multirow{2}{*}{$\lfloor \frac{n}{2}\rfloor$} &\multirow{2}{*}{\cite{Abrego2004, Abrego2009}}&
	    $\lceil\frac{n}{5}\rceil$ & \cite{Abrego2004, Abrego2009}       \\ 
		& &	& $\lceil\frac{n-1}{4}\rceil$& Theorem~\ref{infty-Delaunay-thr} \\ \hline
    \end{tabular}
    \end{minipage}
\end{table}

\subsection{Our results}
\label{our-results}
In this paper we consider the problem of computing a strong matching in $\G{S}{P}$, where $S\in \{\ddisc,\trid, \sqr\}$. In Section~\ref{preliminaries}, we provide some observations and prove necessary Lemmas. Given a point set $P$ in which is in general position with respect to a given shape $S$, in Section~\ref{algorithm-section}, we present an algorithm which computes a strong matching in $\G{S}{P}$. In Section~\ref{Gabriel-section}, we prove that if $S$ is a diametral-disk, then the algorithm of Section~\ref{algorithm-section} computes a strong matching of size at least $\lceil(n-1)/17\rceil$ in $\G{\ddisc}{P}$. In Section~\ref{half-theta-six-section}, we prove that if $S$ is an equilateral triangle, then the algorithm of Section~\ref{algorithm-section} computes a strong matching of size at least $\lceil(n-1)/9\rceil$ in $\G{\trids}{P}$. In Section~\ref{theta-six-section}, we compute a strong matching of size at least $\lceil (n-1)/4\rceil$ in $\GUD(P)$. In Section~\ref{infty-Delaunay-section}, we compute a strong matching of size at least $\lceil(n-1)/4\rceil$ in $\G{\sqrs}{P}$; this improves the previous lower bound of $\lceil n/5\rceil$. A summary of the results is given in Table~\ref{table1}. In Section~\ref{conjecture-section} we discuss a possible way to further improve upon the result obtained for diametral-disks in Section~\ref{Gabriel-section}. Concluding remarks and open problems are given in Section~\ref{conclusion}.

\section{Preliminaries}
\label{preliminaries}

Let $S\in\{\ddisc, \trid, \sqr\}$, and let $S_1$ and $S_2$ be two homothets of $S$. We say that $S_1$ is {\em smaller then} $S_2$ if the area of $S_1$ is smaller than the area of $S_2$. For two points $p,q\in P$, let $S(p,q)$ be a smallest homothet of $S$ having $p$ and $q$ on its boundary. If $S$ is a diametral-disk, a downward equilateral-triangle, or a square, then we denote $S(p,q)$ by $D(p,q)$, $t(p,q)$, or $Q(p,q)$, respectively. If $S$ is a diametral-disk, then $D(p,q)$ is uniquely defined by $p$ and $q$. If $S$ is an equilateral-triangle or a square, then $S$ has the {\em shrinkability} property: if there exists a homothet $S'$ of $S$ that contains two points $p$ and $q$, then there exists a homothet $S''$ of $S$ such that $S''\subseteq S'$, and $p$ and $q$ are on the boundary of $S''$. If $S$ is an equilateral-triangle, then we can
shrink $S''$ further, such that each side of $S''$ contains either $p$ or $q$. If $S$ is a square, then we can
shrink $S''$ further, such that $p$ and $q$ are on opposite sides of $S''$. Thus, we have the following observation:

\begin{observation}
\label{shrink-triangle-obs}
For two points $p,q\in P$,
\begin{itemize}
 \item $D(p,q)$ is uniquely defined by $p$ and $q$, and it has the line segment $pq$ as a diameter.
\item $t(p, q)$ is uniquely defined by $p$ and $q$, and it has one of $p$ and $q$ on a corner and the other point is
on the side opposite to that corner.
\item $Q(p,q)$ has $p$ and $q$ on opposite sides.
\end{itemize}
\end{observation}

\begin{figure}[htb]
  \centering
\setlength{\tabcolsep}{0in}
  $\begin{tabular}{ccc}
\multicolumn{1}{m{.33\columnwidth}}{\centering\includegraphics[width=.22\columnwidth]{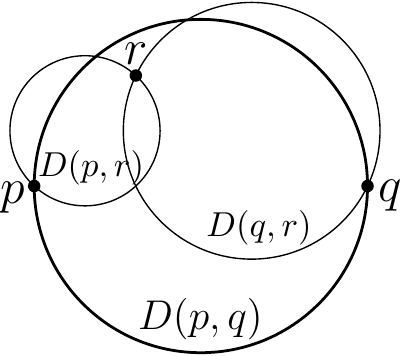}}
&\multicolumn{1}{m{.33\columnwidth}}{\centering\includegraphics[width=.2\columnwidth]{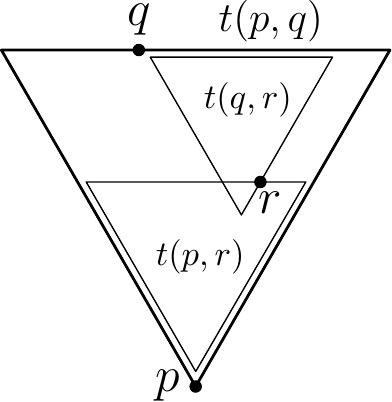}} &\multicolumn{1}{m{.33\columnwidth}}{\centering\includegraphics[width=.21\columnwidth]{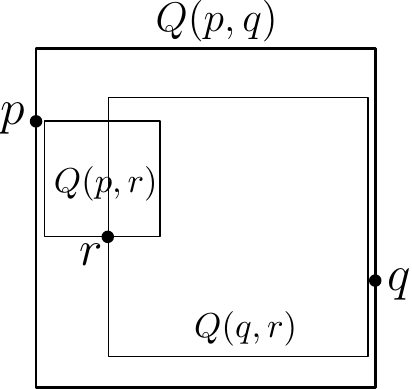}}
\end{tabular}$
  \caption{Illustration of Observation~\ref{obs1}.}
\label{mst-in-GS-fig}
\end{figure}
Given a shape $S\in\{\ddisc, \trid, \sqr\}$, we define an order on the homothets of $S$. Let $S_1$ and $S_2$ be two homothets of $S$. We say that $S_1\prec S_2$ if the area of $S_1$ is less than the area of $S_2$. Similarly, $S_1\preceq S_2$ if the area of $S_1$ is less than or equal to the area of $S_2$. We denote the homothet with the larger area by $\max\{S_1, S_2\}$. As illustrated in Figure~\ref{mst-in-GS-fig}, if $S(p,q)$ contains a point $r$, then both $S(p,r)$ and $S(q,r)$ have smaller area than $S(p,q)$. Thus, we have the following observation:
 
\begin{observation}
\label{obs1}
 If $S(p,q)$ contains a point $r$, then $\max\{S(p,r), \allowbreak S(q,r)\}\prec S(p,q)$.
\end{observation}

\begin{definition}
 Given a point set $P$ and a shape $S\in\{\ddisc, \trid, \sqr\}$, we say that $P$ is in ``general position'' with respect to $S$ if
\begin{description}
 \item[$S=\ddisc$:] no four points of $P$ lie on the boundary of any diametral disk defined by any two points of $P$.
 \item[$S=\trid$:] the line passing through any two points of $P$ does not make angles $0^\circ$, $60^\circ$, or $120^\circ$ with the horizontal. This implies that no four points of $P$ are on the boundary of any homothet of $\trid$.
 \item[$S=\sqr$:] (i) no two points in $P$ have the same $x$-coordinate or the same $y$-coordinate, and (ii) no four points of $P$ lie on the boundary of any homothet of $\sqr$.
\end{description}
\end{definition}

Given a point set $P$ which is in general position with respect to a given shape $S\in\{\ddisc, \trid, \sqr\}$, let $K_S(P)$ be the complete edge-weighted geometric graph on $P$. For each edge $e=(p,q)$ in $K_S(P)$, we define $S(e)$ to be the shape $S(p,q)$, i.e., a smallest homothet of $S$ having $p$ and $q$ on its boundary. We say that $S(e)$ {\em represents} $e$, and vice versa. Furthermore, we assume that the weight $w(e)$ (resp. $w(p,q)$) of $e$ is equal to the area of $S(e)$. Thus,
$$w(p,q)<w(r,s) \quad\text{ if and only if }\quad S(p,q)\prec S(r,s).$$
Note that $\G{S}{P}$ is a subgraph of $K_S(P)$, and has an edge $(p,q)$ iff $S(p,q)$ does not contain any point of $P\setminus\{p,q\}$.

\begin{lemma}
\label{mst-in-GS}
Let $P$ be a set of $n$ points in the plane which is in general position with respect to a given shape $S\in\{\ddisc, \trid, \sqr\}$. Then, any minimum spanning tree of $K_S(P)$ is a subgraph of $\G{S}{P}$. 
\end{lemma}
\begin{proof}
 The proof is by contradiction. Assume there exists an edge $e=(p,q)$ in a minimum spanning tree $T$ of $K_S(P)$ such that $e\notin \G{S}{P}$. Since $(p,q)$ is not an edge in $\G{S}{P}$, $S(p,q)$ contains a point $r$ such that $r\in P\setminus\{p,q\}$. By Observation~\ref{obs1}, $\max\{S(p,r),S(q,r)\}\prec S(p,q)$. Thus, $w(p,r)<w(p,q)$ and $w(q,r)<w(p,q)$. By replacing the edge $(p,q)$ in $T$ with either $(p,r)$ or $(q,r)$, we obtain a spanning tree in $K_S(P)$ which is smaller than $T$. This contradicts the minimality of $T$.
\end{proof}

\begin{lemma}
\label{cycle-lemma}
Let $G$ be an edge-weighted graph with edge set $E$ and edge-weight function $w:E\rightarrow\mathbb{R^+}$. For any cycle $C$ in $G$, if the maximum-weight edge in $C$ is unique, then that edge is not in any minimum spanning tree of $G$.
\end{lemma}

\begin{proof}
The proof is by contradiction. Let $e=(u,v)$ be the unique maximum-weight edge in a cycle $C$ in $G$, such that $e$ is in a minimum spanning tree $T$ of $G$. Let $T_u$ and $T_v$ be the two trees obtained by removing $e$ from $T$. Let $e'=(x,y)$ be an edge in $C$ which connects a vertex $x\in T_u$ to a vertex $y\in T_v$. By assumption, $w(e')< w(e)$. Thus, in $T$, by replacing $e$ with $e'$, we obtain a tree $T'=T_u\cup T_v \cup\{(x,y)\}$ in $G$ such that $w(T')<w(T)$. This contradicts the minimality of $T$. 
\end{proof}

Recall that $t(p,q)$ is the smallest homothet of $\trid$ which has $p$ and $q$ on its boundary. Similarly, let $t'(p, q)$ denote the smallest upward equilateral-triangle $\triu$ having $p$ and $q$ on its boundary. Note that $t'(p, q)$ is uniquely defined by $p$ and $q$, and it has one of $p$ and $q$ on a corner and the other point is on the side opposite to that corner. In addition the area of $t'(p, q)$ is equal to the area of $t(p, q)$. 

\begin{wrapfigure}{r}{0.35\textwidth}
\vspace{-20pt}
 \begin{center}
\includegraphics[width=.3\textwidth]{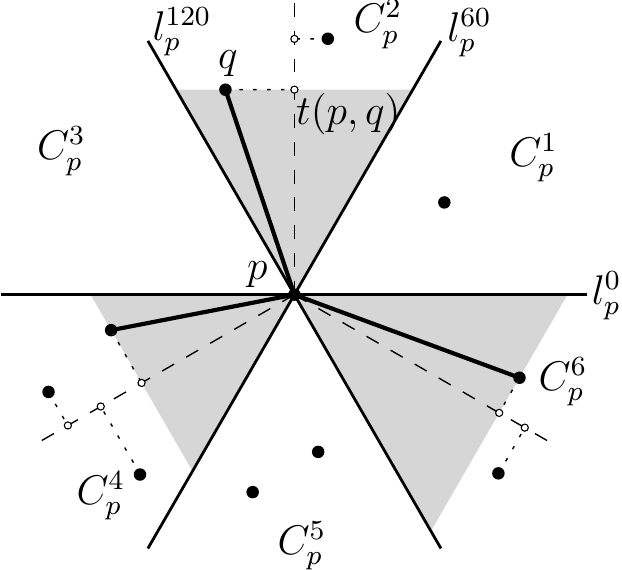}
  \end{center}
\vspace{-5pt}
  \caption{The construction of $\G{\trid}{P}$.}
\label{cones}
\vspace{-5pt}
\end{wrapfigure}

$\G{\trids}{P}$ is equal to the triangular-distance Delaunay graph $\text{\em TD-DG}(P)$, which is in turn equal to a half theta-six graph $\frac{1}{2}\Theta_6(P)$~\cite{Bonichon2010}. 
A half theta-six graph on $P$, and equivalently $\G{\trids}{P}$, can be constructed in the following way. For each point $p$ in $P$, let $l_p$ be the horizontal line through $p$. Define $l_p^{\gamma}$ as the line obtained by rotating $l_p$ by $\gamma$-degrees in counter-clockwise direction around $p$. Thus, $l_p^0=l_p$. Consider three lines $l_p^{0}$, $l_p^{60}$, and $l_p^{120}$ which partition the plane into six disjoint cones with apex $p$. Let $C_p^1, \dots, C_p^6$ be the cones in counter-clockwise order around $p$ as shown in Figure~\ref{cones}. $C_p^1,C_p^3,C_p^5$ will be referred to as {\em odd cones}, and $C_p^2,C_p^4,C_p^6$ will be referred to as {\em even cones}. For each even cone $C_p^i$, connect $p$ to the ``nearest'' point $q$ in $C_p^i$. The {\em distance} between $p$ and $q$, is defined as the Euclidean distance between $p$ and the orthogonal projection of $q$ onto the bisector of $C_p^i$. See Figure~\ref{cones}. In other words, the nearest point to $P$ in $\cone{i}{p}$ is a point $q$ in $\cone{i}{p}$ which minimizes the area of $t(p,q)$. The resulting graph is the half theta-six graph which is defined by even cones \cite{Bonichon2010}. Moreover, the resulting graph is $\G{\trids}{P}$ which is defined with respect to the homothets of $\trid$. By considering the odd cones, $\G{\trius}{P}$ is obtained. By considering the odd cones and the even cones, $\GUD(P)$\textemdash which is equal to $\Theta_6(P)$\textemdash is obtained. Note that $\GUD(P)$ is the union of $\G{\trids}{P}$ and $\G{\trius}{P}$. 

Let $X(p,q)$ be the regular hexagon centered at $p$ which has $q$ on its boundary, and its sides are parallel to $l_p^0$, $l_p^{60}$, and $l_p^{120}$. Then, we have the following observation:
\begin{observation}
\label{obs2}
If $X(p,q)$ contains a point $r$, then $t(p,r)\prec t(p,q)$.
\end{observation}

\section{Strong Matching in $\G{S}{P}$}
\label{algorithm-section}

Given a point set $P$ in the plane which is in general position with respect to a given shape $S\in\{\ddisc, \trid, \sqr\}$, in this section we present an algorithm which computes a strong matching in $\G{S}{P}$. Recall that $K_S(P)$ is the complete edge-weighted graph on $P$ with the weight of each edge $e$ is equal to the area of $S(e)$, where $S(e)$ is a smallest homothet of $S$ representing $e$. Let $T$ be a minimum spanning tree of $K_S(P)$. By Lemma~\ref{mst-in-GS}, $T$ is a subgraph of $\G{S}{P}$. For each edge $e\in T$ we denote by $T(e^+)$ the set of all edges in $T$ whose weight is at least $w(e)$. Moreover, we define the {\em influence set} of $e$, as the set of all edges in $T(e^+)$ whose representing shapes overlap with $S(e)$, i.e.,
$$\Inf{e}=\{e': e'\in T(e^+), S(e')\cap S(e)\neq \emptyset\}.$$

Note that $\Inf{e}$ is not empty, as $e\in \Inf{e}$. Consequently, we define the {\em influence number} of $T$ to be the maximum size of a set among the influence sets of edges in $T$, i.e.,
$$\Inf{T}=\max\{|\Inf{e}|: e\in T\}.$$

Algorithm~\ref{alg1} receives $\G{S}{P}$ as input and computes a strong matching in $\G{S}{P}$ as follows. The algorithm starts by computing a minimum spanning tree $T$ of $\G{S}{P}$, where the weight of each edge is equal to the area of its representing shape. Then it initializes a forest $F$ by $T$, and a matching $\mathcal{M}$ by an empty set. Afterwards, as long as $F$ is not empty, the algorithm adds to $\mathcal{M}$, the smallest edge $\emin$ in $F$, and removes the influence set of $e$ from $F$. Finally, it returns $\mathcal{M}$.
\begin{algorithm}                      
\caption{\SMGG$(\G{S}{P})$}          
\label{alg1} 
\begin{algorithmic}[1]
      \State $T\gets \MST(G_S(P))$
      \State $F\gets T$
      \State $\mathcal{M}\gets \emptyset$
      \While {$F\neq \emptyset$}
	  \State $\emin\gets $ smallest edge in $F$
	  \State $\mathcal{M}\gets \mathcal{M}\cup \{\emin\}$
	  \State $F\gets F - \Inf{\emin}$
	  \EndWhile
    \State \Return $\mathcal{M}$
\end{algorithmic}
\end{algorithm}
\begin{theorem}
\label{GS-thr}
Given a set $P$ of $n$ points in the plane and a shape $S\in\{\ddisc, \trid, \sqr\}$, Algorithm~\ref{alg1} computes a strong matching of size at least $\lceil\frac{n-1}{\emph{Inf}(T)}\rceil$ in $\G{S}{P}$, where $T$ is a minimum spanning tree of $\G{S}{P}$. 
\end{theorem}
\begin{proof}
Let $\mathcal{M}$ be the matching returned by Algorithm~\ref{alg1}. First we show that $\mathcal{M}$ is a strong matching. If $\mathcal{M}$ contains one edge, then trivially, $\mathcal{M}$ is a strong matching. Consider any two edges $e_1$ and $e_2$ in $\mathcal{M}$. Without loss of generality assume that $e_1$ is considered before $e_2$ in the {\sf while} loop. At the time $e_1$ is added to $\mathcal{M}$, the algorithm removes from $F$, the edges in $\Inf{e_1}$, i.e., all the edges whose representing shapes intersect $S(e_1)$. Since $e_2$ remains in $F$ after the removal of $\Inf{e_1}$, $e_2\notin\Inf{e_1}$. This implies that $S(e_1)\cap S(e_2)=\emptyset$, and hence $\mathcal{M}$ is a strong matching.

In each iteration of the {\sf while} loop we select $\emin$ as the smallest edge in $F$, where $F$ is a subgraph of $T$. Then, all edges in $F$ have weight at least $w(e)$. Thus, $F\subseteq T(\emin^+)$; which implies that the set of edges in $F$ whose representing shapes intersect $S(\emin)$ is a subset of $\Inf{\emin}$. Therefore, in each iteration of the {\sf while} loop, out of at most $|\Inf{e}|$-many edges of $T$, we add one edge to $\mathcal{M}$. Since $|\Inf{\emin}|\le \Inf{T}$ and $T$ has $n-1$ edges, we conclude that $|\mathcal{M}|\ge\lceil\frac{n-1}{\Inf{T}}\rceil$.
\end{proof}

\paragraph{Remark}
Let $T$ be the minimum spanning tree computed by Algorithm~\ref{alg1}. Let $e=(u,v)$ be an edge in $T$. Recall that $T(e^+)$ contains all the edges of $T$ whose weight is at least $w(e)$. We define the {\em degree} of $e$ as $\dg{e}=\dg{u}+\dg{v}-1$, where $\dg{u}$ and $\dg{v}$ are the number of edges incident on $u$ and $v$ in $T(e^+)$, respectively. Note that all the edges incident on $u$ or $v$ in $T(e^+)$ are in the influence set of $e$. Thus, $|\Inf{e}|\ge \dg{e}$, and consequently $\Inf{T}\ge \dg{e}$.

\section{Strong Matching in $\G{\ddiscs}{P}$}
\label{Gabriel-section}
In this section we consider the case where $S$ is a diametral-disk $\ddisc$. Recall that $\G{\ddiscs}{P}$ is an edge-weighted geometric graph, where the weight of an edge $(p,q)$ is equal to the area of $D(p,q)$. $\G{\ddiscs}{P}$ is equal to the Gabriel graph, $GG(P)$. We prove that $\G{\ddiscs}{P}$, and consequently $GG(P)$, has a strong diametral-disk matching of size at least $\lceil\frac{n-1}{17}\rceil$. 

We run Algorithm~\ref{alg1} on $\G{\ddiscs}{P}$ to compute a matching $\mathcal{M}$. By Theorem~\ref{GS-thr}, $\mathcal{M}$ is a strong matching of size at least $\lceil\frac{n-1}{\Inf{T}}\rceil$, where $T$ is a minimum spanning tree in $\G{\ddiscs}{P}$. By Lemma~\ref{mst-in-GS}, $T$ is a minimum spanning tree of the complete graph $K_{\ddiscs}(P)$. Observe that $T$ is a Euclidean minimum spanning tree for $P$ as well. In order to prove the desired lower bound, we show that $\Inf{T}\le 17$. Since $\Inf{T}$ is the maximum size of a set among the
influence sets of edges in $T$, it suffices to show that for every edge $e$ in $T$, the influence set of $e$ contains at most 17 edges. 
\begin{lemma}
\label{disk-inf-lemma}
Let $T$ be a minimum spanning tree of $\G{\ddiscs}{P}$, and let $e$ be any edge in $T$. Then, $|\emph{Inf}(e)|\le 17$.
\end{lemma}
We will prove this lemma in the rest of this section. Recall that, for each two points $p,q\in P$, $D(p,q)$ is the closed diametral-disk with diameter $pq$. Let $\mathcal{D}$ denote the set of diametral-disks representing the edges in $T$. Since $T$ is a subgraph of $\G{\ddiscs}{P}$, we have the following observation:

\begin{observation}
\label{no-point-in-circle-obs}
 Each disk in $\mathcal{D}$ does not contain any point of $P$ in its interior.
\end{observation}

Recall that, for each two points $p,q\in P$, $D(p,q)$ is the closed diametral-disk with diameter $pq$. Let $\mathcal{D}$ denote the set of diametral-disks representing the edges in $T$. Since $T$ is a subgraph of $\G{\ddiscs}{P}$, we have the following observation:

\begin{observation}
\label{no-point-in-circle-obs}
 Each disk in $\mathcal{D}$ does not contain any point of $P$ in its interior.
\end{observation}

\begin{lemma}
\label{center-in-lemma}
 For each pair $D_i$ and $D_j$ of disks in $\mathcal{D}$, $D_i$ (resp. $D_j$) does not contain the center of $D_j$ (resp $D_i$).
\end{lemma}

\begin{proof}
 Let $(a_i,b_i)$ and $(a_j,b_j)$ respectively be the edges of $T$ which correspond to $D_i$ and $D_j$. Let $C_i$ and $C_j$ be the circles representing the boundary of $D_i$ and $D_j$. W.l.o.g. assume that $C_j$ is the bigger circle, i.e., $|a_ib_i|<|a_jb_j|$. By contradiction, suppose that $C_j$ contains the center $c_i$ of $C_i$. Let $x$ and $y$ denote the intersections of $C_i$ and $C_j$. Let $x_i$ (resp. $x_j$) be the intersection of $C_i$ (resp. $C_j$) with the line through $y$ and $c_i$ (resp. $c_j$). Similarly, let $y_i$ (resp. $y_j$) be the intersection of $C_i$ (resp. $C_j$) with the line through $x$ and $c_i$ (resp. $c_j$). 

\begin{figure}[htb]
  \centering
  \includegraphics[width=.6\columnwidth]{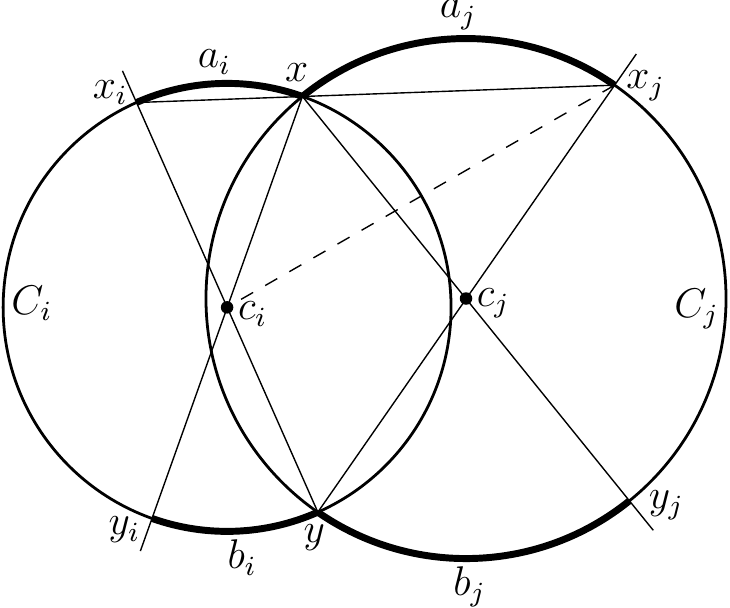}
 \caption{Illustration of Lemma~\ref{center-in-lemma}: $C_i$ and $C_j$ intersect, and $C_j$ contains the center of $C_i$.}
  \label{center-in-fig}
\end{figure}

As illustrated in Figure~\ref{center-in-fig}, the arcs $\widehat{x_ix}$, $\widehat{y_iy}$, $\widehat{x_jx}$, and $\widehat{y_jy}$ are the potential positions for the points $a_i$, $b_i$, $a_j$, and $b_j$, respectively. First we will show that the line segment $x_ix_j$ passes through $x$ and $|a_ia_j|\leq|x_ix_j|$. The angles $\angle x_ixy$ and $\angle x_jx_y$ are right angles, thus the line segment $x_ix_j$ goes through $x$. Since $\widehat{x_ix}<\pi$ (resp. $\widehat{x_jx}<\pi$), for any point $a_i\in \widehat{x_ix}, |a_ix|\leq|x_ix|$ (resp. $a_j\in \widehat{x_jx}, |a_jx|\leq|x_jx|$). Therefore, $$|a_ia_j|\leq|a_ix|+|xa_j|\leq|x_ix|+|xx_j|=|x_ix_j|.$$
Consider triangle $\bigtriangleup x_ix_jy$ which is partitioned by segment $c_ix_j$ into $t_1=\bigtriangleup x_ix_jc_i$ and $t_2=\bigtriangleup c_ix_jy$. It is easy to see that $|x_ic_i|$ in $t_1$ is equal to $|c_iy|$ in $t_2$, and the segment $c_ix_j$ is shared by $t_1$ and $t_2$. Since $c_i$ is inside $C_j$ and $\widehat{yx_j}=\pi$, the angle $\angle yc_ix_j>\frac{\pi}{2}$. Thus, $\angle x_ic_ix_j$ in $t_1$ is smaller than $\frac{\pi}{2}$ (and hence smaller than $\angle yc_ix_j$ in $t_2$). That is,  $|x_ix_j|$ in $t_1$ is smaller than $|x_jy|$ in $t_2$. Therefore,

$$|a_ia_j|\leq|x_ix_j|<|x_jy|=|a_jb_j|.$$

By symmetry $|b_ib_j|<|a_jb_j|$. Therefore $\max\{|a_ia_j|,|b_ib_j|\}<\max\{|a_ib_i|,|a_jb_j|\}$. Therefore, the cycle $a_i,a_j,b_j,b_i,a_i$ contradicts Lemma~\ref{cycle-lemma}.
\end{proof}

Let $e=(u,v)$ be an edge in $T$. Without loss of generality, we suppose that $D(u,v)$ has radius 1 and centered at the origin $o=(0,0)$ such that $u=(-1,0)$ and $v=(1,0)$. For any point $p$ in the plane, let $\|p\|$ denote the distance of $p$ from $o$. Let $\mathcal{D}(e^+)$ be the disks in $\mathcal{D}$ representing the edges of $T(e^+)$. Recall that $T(e^+)$ contains the edges of $T$ whose weight is at least $w(e)$, where $w(e)$ is equal to the area of $\cmin$. Since the area of any circle is directly related to its radius, we have the following observation:

\begin{observation}
 \label{radius-one}
The disks in $\mathcal{D}(e^+)$ have radius at least $1$.
\end{observation}

Let $C(x,r)$ (resp. $D(x,r)$) be the circle (resp. closed disk) of radius $r$ which is centered at a point $x$ in the plane. 
Let $\mathcal{I}(e^+)=\{D_1,\dots, D_k\}$ be the set of disks in $\mathcal{D}(e^+)\setminus\{D(u,v)\}$ intersecting $D(u,v)$. We show that $\mathcal{I}(e^+)$ contains at most sixteen disks, i.e., $k\le 16$.

For $i\in\{1,\dots,k\}$, let $c_i$ denote the center of the disk $D_i$. 
In addition, let $c'_i$ be the intersection point between $C(o,2)$ and the ray with origin at $o$ which passing through $c_i$. Let the point $p_i$ be $c_i$, if $\|c_i\|< 2$, and $c'_i$, otherwise. See Figure~\ref{distance-fig}. Finally, let $P'=\{o, u, v, p_1,\dots,p_k\}$. 

\begin{observation}
\label{obs}
Let $c_j$ be the center of a disk $D_j$ in $\mathcal{I}(e^+)$, where $\|c_j\|\ge 2$. Then, the disk $D(c_j, \|c_j\|-1)$ is contained in the disk $D_j$. Moreover, the disk $D(p_j,1)$ is contained in the disk $D(c_j, \|c_j\|-1)$. See Figure~\ref{distance-fig}.
\end{observation}

\begin{figure}[htb]
  \centering
  \includegraphics[width=.5\columnwidth]{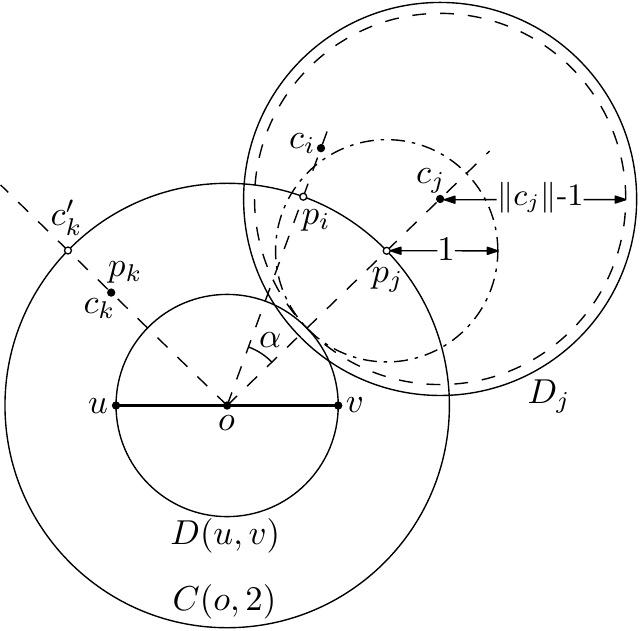}
 \caption{Proof of Lemma~\ref{distance-lemma}; $p_i=c'_i$, $p_j=c'_j$, and $p_k=c_k$.}
  \label{distance-fig}
\end{figure}

\begin{lemma}
\label{distance-lemma}
The distance between any pair of points in $P'$ is at least 1.
\end{lemma}
\begin{proof}
Let $x$ and $y$ be two points in $P'$. We are going to prove that $|xy|\ge 1$. We distinguish between the following three cases. 
\begin{itemize}
 \item $x,y \in\{o,u,v\}$. In this case the claim is trivial.
\item $x\in \{o,u,v\}, y\in \{p_1,\dots, p_k\}$. If $\|y\|=2$, then $y$ is on $C(o,2)$, and hence $|xy|\ge 1$. If $\|y\|<2$, then $y$ is the center of a disk $D_i$ in $\mathcal{I}(e^+)$. By Observation~\ref{no-point-in-circle-obs}, $D_i$ does not contain $u$ and $v$, and by Lemma~\ref{center-in-lemma}, $D_i$ does not contain $o$. Since $D_i$ has radius at least 1, we conclude that $|xy|\ge 1$.

\item $x,y\in\{p_1,\dots,p_k\}$. Without loss of generality assume $x=p_i$ and $y= p_j$, where $1\le i<j\le k$. We differentiate between three cases:
\begin{itemize}
 \item $\|p_i\|< 2$ and $\|p_j\|<2$. In this case $p_i$ and $p_j$ are the centers of $D_i$ and $D_j$, respectively. By Lemma~\ref{center-in-lemma} and Observation~\ref{radius-one}, we conclude that $|p_ip_j|\ge 1$.

\item $\|p_i\|< 2$ and $\|p_j\|=2$. By Observation~\ref{obs} the disk $D(p_j, 1)$ is contained in the disk $D_j$. By Lemma~\ref{center-in-lemma}, $p_i$ is not in the interior of $D_j$, and consequently, it is not in the interior of $D(p_j,1)$. Therefore, $|p_ip_j|\ge 1$.
\item $\|p_i\|= 2$ and $\|p_j\|=2$. Recall that $c_i$ and $c_j$ are the centers of $D_i$ and $D_j$, such that $\|c_i\|\ge 2$ and $\|c_j\|\ge2$. Without loss of generality assume $\|c_i\|\le \|c_j\|$. For the sake of contradiction assume that $|p_ip_j|<1$. Then, for the angle $\alpha=\angle c_i o c_j$ we have $\sin(\alpha/2)< \frac{1}{4}$. Then, $\cos(\alpha)> 1-2\sin^2(\alpha/2)=\frac{7}{8}$. By the law of cosines in the triangle $\bigtriangleup c_ioc_j$, we have
\begin{equation}
\label{ineq1}
|c_ic_j|^2<\|c_i\|^2+\|c_j\|^2-\frac{14}{8}\|c_i\|\|c_j\|.
\end{equation}
By Observation~\ref{obs} the disk $D(c_j,\|c_j\|-1)$ is contained in $D_j$; see Figure~\ref{distance-fig}. By Lemma~\ref{center-in-lemma}, $c_i$ is not in the interior of $D_j$, and consequently, is not in the interior of $D(c_j,\|c_j\|-1)$. Thus, $|c_ic_j|\ge \|c_j\|-1$. In combination with Inequality~(\ref{ineq1}), this gives
\begin{equation}
 \label{ineq2}
\|c_j\|\left(\frac{14}{8}\|c_i\|-2\right) < \|c_i\|^2-1.
\end{equation}

In combination with the assumption that $\|c_i\| \le \|c_j\|$, Inequality~(\ref{ineq2}) gives
$$\frac{6}{8}\|c_i\|^2-2\|c_i\|+1<0.$$

To satisfy this inequality, we should have $\|c_i\|<2$, contradicting the fact
that $\|c_i\| \ge 2$. This completes the proof.
\end{itemize}
\end{itemize}
\end{proof}

By Lemma~\ref{distance-lemma}, the points in $P'$ has mutual distance 1. Moreover, the points in $P'$ lie in (including the boundary) $C(o,2)$.
Bateman and Erd\H{o}s~\cite{Bateman1951} proved that it is impossible to have 20 points in (including the boundary) a circle of radius 2 such that one of the points is at the center and all of the mutual distances are at least 1.
Therefore, $P'$ contains at most $19$ points, including $o$, $u$, and $v$. This implies that $k\le 16$, and hence $\mathcal{I}(e^+)$ contains at most sixteen edges. This completes the proof of Lemma~\ref{disk-inf-lemma}.

\begin{theorem}
 \label{Gabriel-thr}
Algorithm~\ref{alg1} computes a strong matching of size at least $\lceil\frac{n-1}{17}\rceil$ in $\G{\ominus}{P}$.
\end{theorem}

\section{Strong Matching in $\G{\trids}{P}$}
\label{half-theta-six-section}
In this section we consider the case where $S$ is a downward equilateral triangle $\trid$, whose barycenter is the origin and one of its vertices is on the negative $y$-axis. In this section we assume that $P$ is in general position, i.e., for each point $p\in P$, there is no point of $P\setminus \{p\}$ on $l_p^0$, $l_p^{60}$, and $l_p^{120}$. In combination with Observation~\ref{shrink-triangle-obs}, this implies that for two points $p,q\in P$, no point of $P\setminus\{p,q\}$ are on the boundary of $t(p,q)$ (resp. $t'(p,q)$). Recall that $t(p,q)$ is the smallest homothet of $\trid$ having of $p$ and $q$ on a corner and the other point on the side opposite to that corner. We prove that $\G{\trids}{P}$, and consequently $\frac{1}{2}\Theta_6(P)$, has a strong triangle matching of size at least $\lceil\frac{n-1}{9}\rceil$. 

We run Algorithm~\ref{alg1} on $\G{\trids}{P}$ to compute a matching $\mathcal{M}$. Recall that $\G{\trids}{P}$ is an edge-weighted graph with the weight of each edge $(p,q)$ is equal to the area of $t(p,q)$. By Theorem~\ref{GS-thr}, $\mathcal{M}$ is a strong matching of size at least $\lceil\frac{n-1}{\Inf{T}}\rceil$, where $T$ is a minimum spanning tree in $\G{\trids}{P}$. In order to prove the desired lower bound, we show that $\Inf{T}\le 9$. Since $\Inf{T}$ is the maximum size of a set among the
influence sets of edges in $T$, it suffices to show that for every edge $e$ in $T$, the influence set of $e$ has at most nine edges. 
\begin{lemma}
\label{triangle-inf-lemma}
Let $T$ be a minimum spanning tree of $\G{\trids}{P}$, and let $e$ be any edge in $T$. Then, $|\emph{Inf}(e)|\le 9$.
\end{lemma}

 \begin{figure}[htb]
  \centering
\setlength{\tabcolsep}{0in}
  $\begin{tabular}{ccc}
\multicolumn{1}{m{.33\columnwidth}}{\centering\includegraphics[width=.2\columnwidth]{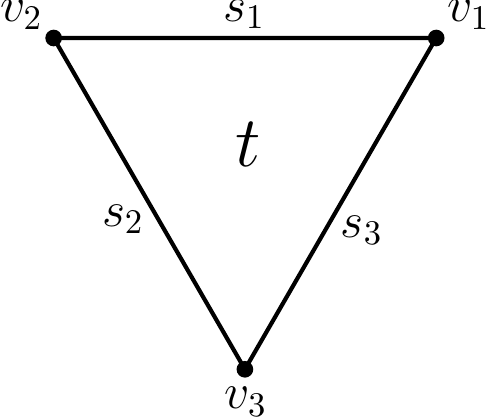}}
&\multicolumn{1}{m{.33\columnwidth}}{\centering\includegraphics[width=.2\columnwidth]{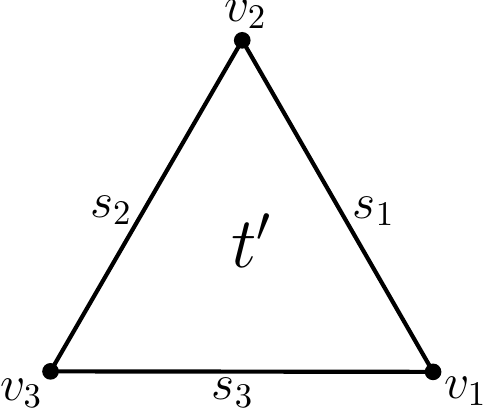}} 
&\multicolumn{1}{m{.33\columnwidth}}{\centering\includegraphics[width=.2\columnwidth]{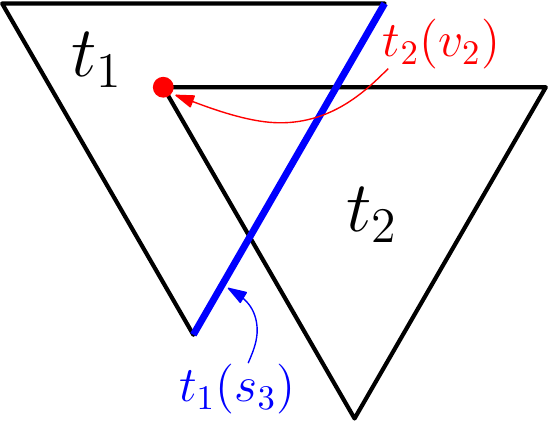}}\\
(a)&(b)&(c)
\end{tabular}$
  \caption{(a) Labeling the vertices and the sides of a downward triangle. (b) Labeling the vertices and the sides of an upward triangle. (c) Two intersecting triangles.}
  \label{triangle-fig}
\end{figure}

We will prove this lemma in the rest of this section. We label the vertices and the sides of a downward equilateral-triangle, $t$, and an upward equilateral-triangle, $t'$, as depicted in Figures~\ref{triangle-fig}(a) and ~\ref{triangle-fig}(b). We refer to a vertex $v_i$ and a side $s_i$ of a triangle $t$ by $t(v_i)$ and $t(s_i)$, respectively.

Recall that $F$ is a subgraph of the minimum spanning tree $T$ in $\G{\trids}{P}$. In each iteration of the {\sf while} loop in Algorithm~\ref{alg1}, let $\mathcal{T}$ denote the set of triangles representing the edges in $F$. By Lemma~\ref{mst-in-GS} and the general position assumption we have

\begin{observation}
\label{no-point-in-triangle-obs}
Each triangle $t(p,q)$ in $\mathcal{T}$ does not contain any point of $P\setminus \{p,q\}$ in its interior or on its boundary.\vspace{-5pt}
\end{observation}

Consider two intersecting triangles $t_1(p_1,q_1)$ and $t_2(p_2,q_2)$ in $\mathcal{T}$. By Observation~\ref{shrink-triangle-obs}, each side of $t_1$ contains either $p_1$ or $q_1$, and each side of $t_2$ contains either $p_2$ or $q_2$. Thus, by Observation~\ref{no-point-in-triangle-obs}, we argue that no side of $t_1$ is completely in the interior of $t_2$, and vice versa. Therefore, either exactly one vertex (corner) of $t_1$ is in the interior of $t_2$, or exactly one vertex of $t_2$ is in the interior of $t_1$. Without loss of generality assume that a corner of $t_2$ is in the interior of $t_1$, as shown in Figure~\ref{triangle-fig}(c). In this case we say that $t_1$ intersects $t_2$ through the vertex $t_2(v_2)$, or symmetrically, $t_2$ intersects $t_1$ through the side $t_1(s_3)$.

The following two lemmas have been proved by Biniaz et al.~\cite{Biniaz2015}:
\begin{lemma}[Biniaz et al.~\cite{Biniaz2015}]
\label{triangle3}
Let $t_1$ be a downward triangle which intersects a downward triangle $t_2$ through $\tra{t_2}{s_1}$, and let a horizontal line $\ell$ intersects both $t_1$ and $t_2$. Let $p_1$ and $q_1$ be two points on $t_1(s_2)$ and $t_1(s_3)$, respectively, which are above $t_2(s_1)$. Let $p_2$ and $q_2$ be two points on $t_2(s_2)$ and $t_2(s_3)$, respectively, which are above $\ell$. Then, $\max\{t(p_1,p_2), t(q_1,q_2)\} \prec \max\{t_1,t_2\}$. See Figure~\ref{triangle-intersection-fig}(b).
\end{lemma}

\begin{lemma}[Biniaz et al.~\cite{Biniaz2015}]
\label{intersection-lemma}
For every four triangles $t_1,t_2,t_3,t_4\in \mathcal{T}$, $t_1\cap t_2\cap t_3\cap t_4 =\emptyset$. 
\end{lemma}

As a consequence of Lemma~\ref{triangle3}, we have the following corollary:
\begin{corollary}
\label{biniaz-cor}
 Let $t_1, t_2, t_3$ be three triangles in $\mathcal{T}$. Then $t_1$, $t_2$, and $t_3$ cannot make a chain configuration, such that $t_2$ intersects $t_3$ through $t_3(s_1)$ and $t_1$ intersects both $t_2$ and $t_3$ through $t_2(s_1)$ and $t_3(s_1)$. See Figure~\ref{triangle-intersection-fig}(b).
\end{corollary}

\begin{figure}[htb]
  \centering
\setlength{\tabcolsep}{0in}
  $\begin{tabular}{cc}
\multicolumn{1}{m{.5\columnwidth}}{\centering\includegraphics[width=.28\columnwidth]{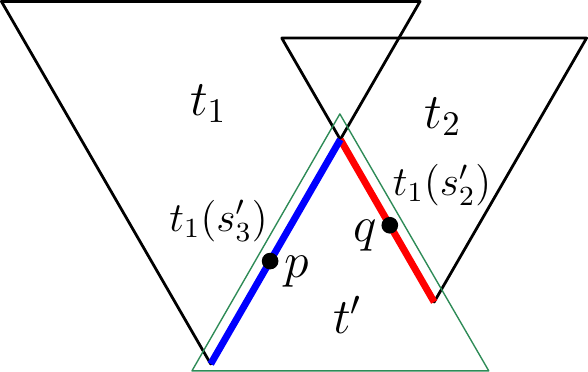}}
&\multicolumn{1}{m{.5\columnwidth}}{\centering\includegraphics[width=.25\columnwidth]{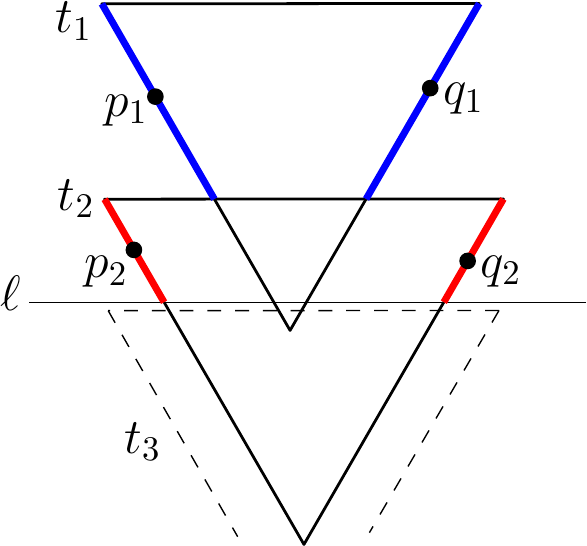}} 
\\(a)&(b)
\end{tabular}$
  \caption{(a) Illustration of Lemma~\ref{triangle-intersection-lemma}. (b) Illustration of Lemma~\ref{triangle3}.}
  \label{triangle-intersection-fig}
\end{figure}

\begin{lemma}
 \label{triangle-intersection-lemma}
Let $t_1$ be a downward triangle which intersects a downward triangle $t_2$ through $t_2(v_2)$. Let $p$ be a point on $t_1(s_3)$ and to the left of $t_2(s_2)$, and let $q$ be a point on $t_2(s_2)$ and to the right of $t_1(s_3)$. Then, $\tr{p,q}\prec\max\{t_1,t_2\}$.
\end{lemma}
\begin{proof}
 Refer to Figure~\ref{triangle-intersection-fig}(a). Let $t_1(s'_3)$ be the part of the line segment $t_1(s_3)$ which is to the left of $t_2(s_2)$, and let $t_2(s'_2)$ be the part of the line segment $t_2(s_2)$ which is to the right of $t_1(s_3)$. Without loss of generality assume that $t_1(s'_3)$ is larger than $t_2(s'_2)$. Let $t'$ be an upward triangle having $t_1(s'_3)$ as its left side. Then, $t'\prec t_1$, which implies that $t'\prec\max\{t_1,t_2\}$. Since $t'$ has both $p$ and $q$ on its boundary, the area of the downward triangle $t(p,q)$ is smaller than the area of $t'$. Therefore, $\tr{p,q}\preceq t'$; which completes the proof.
\end{proof}

Because of the symmetry, the statement of Lemma~\ref{triangle-intersection-lemma} holds even if $p$ is above $t_2(s_1)$ and $q$ is on $t_2(s_1)$.
Consider the six cones with apex at $p$, as shown in Figure~\ref{cones}.
\begin{lemma}
\label{deg-six-half}
Let $T$ be a minimum spanning tree in $\G{\trids}{P}$. Then, in $T$, every point $p$ is adjacent to at most one point in each cone $\cone{i}{p}$, where $1\le i\le 6$.
\end{lemma}
\begin{proof}
If $i$ is even, then by the construction of $\G{\trids}{P}$, which is given in Section~\ref{preliminaries}, $p$ is adjacent to at most one point in $\cone{i}{p}$. Assume $i$ is odd. For the sake of contradiction, assume in $T$, the point $p$ is adjacent to two points $q$ and $r$ in a cone $\cone{i}{p}$. Then, $t(p,q)$ has $q$ on a corner, and $t(p,r)$ has $r$ on a corner. Without loss of generality assume $t(p,r)\prec t(q,r)$. Then, the hexagon $\hex{q}{p}$ has $r$ in its interior. Thus, $\tr{q,r}\prec \tr{p,q}$. Then the cycle $r,p,q,r$ contradicts Lemma~\ref{cycle-lemma}. Therefore, $p$ is adjacent to at most one point in each of the six cones.
\end{proof}

In Algorithm~\ref{alg1}, in each iteration of the {\sf while} loop, let $\mathcal{T}(e^+)$ be the triangles representing the edges of $F$. Recall that $e$ is the smallest edge in $F$, and hence, $t(e)$ is a smallest triangle in $\mathcal{T}(e^+)$.
Let $e=(p,q)$ and let $\mathcal{I}(e^+)$ be the set of triangles in $\mathcal{T}(e^+)$ (excluding $t(e)$) which intersect $t(e)$. We show that $\mathcal{I}(e^+)$ contains at most eight triangles.
We partition the triangles in $\mathcal{I}(e^+)$ into $\{\mathcal{I}_1,\mathcal{I}_2\}$, such that every triangle $\tau\in\mathcal{I}_1$ shares only $p$ or $q$ with $t=t(e)$, i.e., $\mathcal{I}_1=\{\tau: \tau\in\mathcal{I}(e^+),\tau\cap t\in \{p,q\}\}$, and every triangle $\tau\in\mathcal{I}_2$ intersects $t$ either through a side or through corner which is not $p$ nor $q$.

\begin{wrapfigure}{r}{0.4\textwidth}
\vspace{-20pt}
 \begin{center}
\includegraphics[width=.35\textwidth]{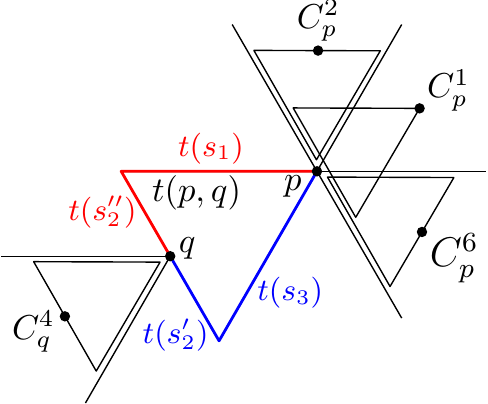}
  \end{center}
\vspace{-15pt}
  \caption{Illustration of the triangles in $\mathcal{I}_1$.}
\label{cones-pq}
\vspace{-8pt}
\end{wrapfigure}
By Observation~\ref{shrink-triangle-obs}, for each triangle $t(p,q)$, one of $p$ and $q$ is on a corner of $t(p,q)$ and the other one is on the side opposite to that corner. Without loss of generality assume that $p$ is on the corner $t(v_1)$, and hence, $q$ is on the side $t(s_2)$. See Figure~\ref{cones-pq}. Note that the other cases, where $p$ is on $t(v_2)$ or on $t(v_3)$ are similar.
Since the intersection of $t$ with any triangle $\tau\in\mathcal{I}_1$ is either $p$ or $q$, $\tau$ has either $p$ or $q$ on its boundary. In combination with Observation~\ref{no-point-in-triangle-obs}, this implies that $\tau$ represent an edge $e'$ in $T$, and hence, either $p$ or $q$ is an endpoint of $e'$. As illustrated in Figure~\ref{cones-pq}, the other endpoint of $e'$ can be either in $\cone{1}{p}$, $\cone{2}{p}$, $\cone{6}{p}$, or in $\cone{4}{q}$, because otherwise $\tau\cap t\notin \{p,q\}$. By Lemma~\ref{deg-six-half}, $p$ has at most one neighbor in each of $\cone{1}{p}$, $\cone{2}{p}$, $\cone{6}{p}$, and $q$ has at most one neighbor in $\cone{4}{q}$. Therefore, $\mathcal{I}_1$ contains at most four triangles. We are going to show that $\mathcal{I}_2$ also contains at most four triangles. 

The point $q$ divides $t(s_2)$ into two parts. Let $\tra{t}{s'_2}$ and $\tra{t}{s''_2}$ be the parts of $\tra{t}{s_2}$ which are below and above $q$, respectively; see Figure~\ref{cones-pq}. The triangles in $\mathcal{I}_2$ intersect $t$ either through $t(s_1)\cup t(s''_2)$ or through $t(s_3)\cup t(s'_2)$; which are shown by red and blue polylines in Figure~\ref{cones-pq}. We show that most two triangles in $\mathcal{I}_2$ intersect $t$ through each of $t(s_1)\cup t(s''_2)$ or $t(s_3)\cup t(s'_2)$. Because of symmetry, we only prove for $t(s_3)\cup t(s'_2)$. When a triangle $t'$ intersects $t$ through both $t(s_3)$ and $t(s'_2)$ we say $t'$ intersects $t$ through $t(v_3)$. In the next lemma, we prove that at most one triangle in $\mathcal{I}_2$ intersects $t$ through each of $\tra{t}{s_3}$, $\tra{t}{s'_2}$. Again, because of symmetry, we only prove for $\tra{t}{s_3}$. 

\begin{figure}[htb]
  \centering
\setlength{\tabcolsep}{0in}
  $\begin{tabular}{cc}
\multicolumn{1}{m{.5\columnwidth}}{\centering\includegraphics[width=.3\columnwidth]{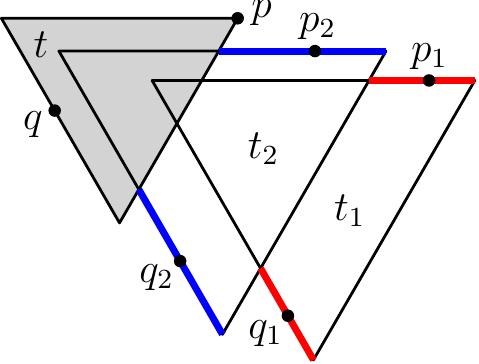}}
&\multicolumn{1}{m{.5\columnwidth}}{\centering\includegraphics[width=.3\columnwidth]{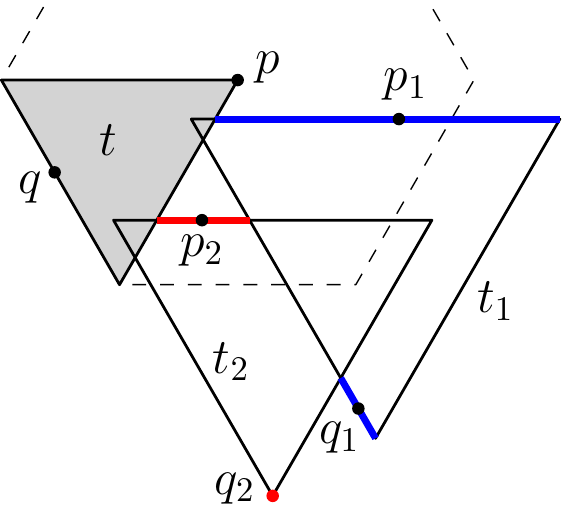}} \\
(a) & (b)
\end{tabular}$
  \caption{Illustration of Lemma~\ref{side-intersection}: (a) $t_1(v_2)\in t_2$. (b) $t_1(v_2)\notin t_2$ and $t_2(v_2)\notin t_1$.}
  \label{two-triangles-fig}
\end{figure}

\begin{lemma}
\label{side-intersection}
At most one triangle in $\mathcal{I}_2$ intersects $\trmin$ through $\tra{\trmin}{s_3}$.
\end{lemma}
\begin{proof}
The proof is by contradiction. Assume two triangles $t_1(p_1,q_1)$ and $t_2(p_2,q_2)$ in $\mathcal{I}_2$ intersect $\trmin$ through $\tra{\trmin}{s_3}$. Without loss of generality assume that $p_i$ is on $t_i(s_1)$ and $q_i$ is on $t_i(s_2)$ for $i=1,2$. Recall that the area of $t_1$ and the area of $t_2$ are at least the area of $\trmin$. If $\tra{t_1}{v_2}$ is in the interior of $t_2$ (as shown in Figure~\ref{two-triangles-fig}(a)) or $\tra{t_2}{v_2}$ is in the interior of $t_1$, then we get a contradiction to Corollary~\ref{biniaz-cor}. Thus, assume that $\tra{t_1}{v_2}\notin t_2$ and $\tra{t_2}{v_2} \notin t_1$. 

Without loss of generality assume that $\tra{t_1}{s_1}$ is above $\tra{t_2}{s_1}$; see Figure~\ref{two-triangles-fig}(b). By Lemma~\ref{triangle-intersection-lemma}, we have $t(p,p_1)\prec \max\{t,t_1\}\preceq t_1$. If $q_1$ is in $\hex{p}{q}$, then by Observation~\ref{obs2}, $t(p,q_1)\prec t$. Then, the cycle $p,p_1,q_1,p$ contradicts Lemma~\ref{cycle-lemma}. Thus, assume that $q_1\notin \hex{p,q}$. In this case $\tra{t_2}{s_3}$ is to the left of $\tra{t_1}{s_3}$, because otherwise $q_1$ lies in $t_2$ which contradicts Observation~\ref{no-point-in-triangle-obs}.
Since both $t_1$ and $t_2$ are larger than $t$, $t_2$ intersects $t_1$ through $t_1(s_2)$, and hence $\tra{t_2}{v_1}$ is in the interior of $t_1$. This implies that $q_2$ is on $\tra{t_2}{v_3}$. In addition, $p_2$ is on the part of $\tra{t_2}{s_1}$ which lies in the interior of $\hex{p}{q}$. By Observation~\ref{obs2} and Lemma~\ref{triangle-intersection-lemma}, we have $t(p,p_2)\prec t$ and $t(q_1,q_2)\prec \max\{t_1,t_2\}$, respectively. Thus, the cycle $p,p_1,q_1,q_2,p_2,p$ contradicts Lemma~\ref{cycle-lemma}. 
\end{proof}

\begin{lemma}
\label{vertex-intersection}
 At most two triangles in $\mathcal{I}_2$ intersect $\trmin$ through $\tra{\trmin}{v_3}$.
\end{lemma}
\begin{proof}
For the sake of contradiction assume three triangles $t_1,t_2, t_3\in \mathcal{I}_2$ intersect $\trmin$ through $\tra{\trmin}{v_3}$. This implies that $\tra{\trmin}{v_3}$ belongs to four triangles $t,t_1,t_2,t_3$, which contradicts Lemma~\ref{intersection-lemma}. \end{proof}

\begin{figure}[htb]
  \centering
\setlength{\tabcolsep}{0in}
  $\begin{tabular}{cc}
\multicolumn{1}{m{.5\columnwidth}}{\centering\includegraphics[width=.35\columnwidth]{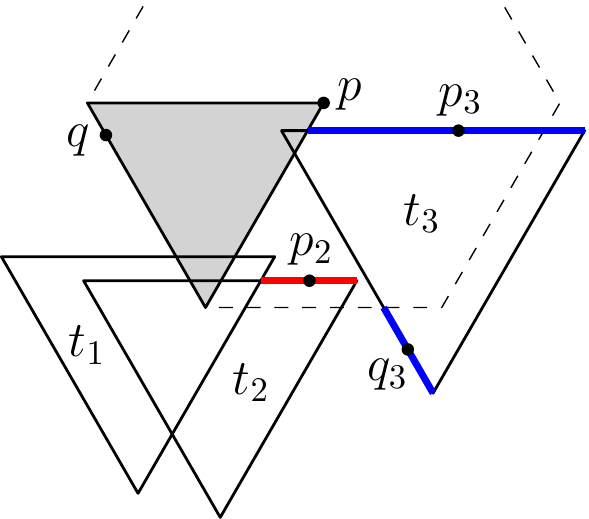}}
&\multicolumn{1}{m{.5\columnwidth}}{\centering\includegraphics[width=.35\columnwidth]{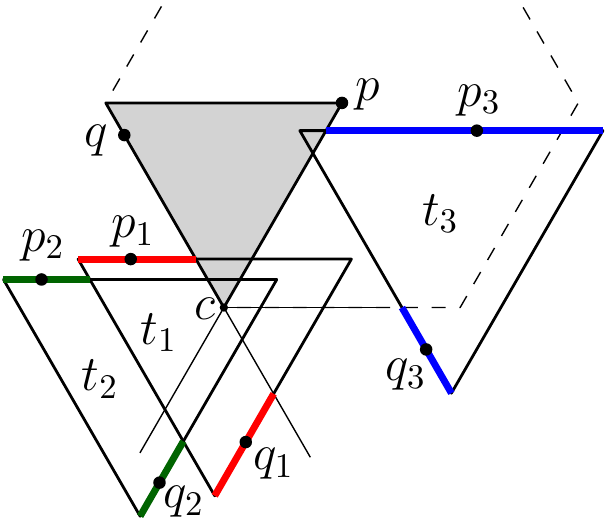}} \\
(a) & (b)\\\multicolumn{1}{m{.5\columnwidth}}{\centering\includegraphics[width=.35\columnwidth]{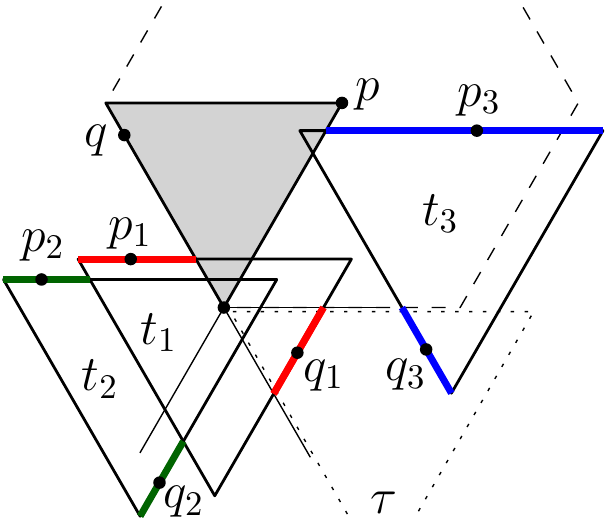}}
&\multicolumn{1}{m{.5\columnwidth}}{\centering\includegraphics[width=.35\columnwidth]{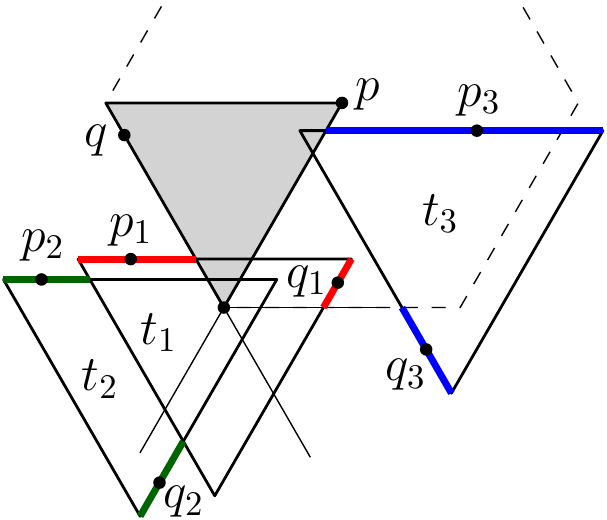}} \\
(c) & (d)
\end{tabular}$
  \caption{Illustration of Lemma~\ref{vertex-side-intersection-2}: (a) $p_2$ is to the right of $t_1(s_3)$, (b) $q_1\in C^5_{\tra{t}{v_3}}$, (c) $q_1\in C^6_{\tra{t}{v_3}}$, and (d) $q_1\in C^1_{\tra{t}{v_3}}$.}
\label{three-triangle-fig2}
\end{figure}
\begin{lemma}
\label{vertex-side-intersection-2}
If two triangles in $\mathcal{I}_2$ intersect $\trmin$ through $\tra{\trmin}{v_3}$, then no other triangle in $\mathcal{I}_2$ intersects $\trmin$ through $\tra{\trmin}{s_3}$ or through $\tra{\trmin}{s'_2}$. 
\end{lemma}

\begin{proof}
The proof is by contradiction. Assume two triangles $t_1(p_1,q_1)$ and $t_2(p_2,q_2)$ in $\mathcal{I}_2$ intersect $\trmin$ through $\tra{\trmin}{v_3}$, and a triangle $t_3(p_3,q_3)$ in $\mathcal{I}_2$ intersects $\trmin$ through $\tra{\trmin}{s_3}$ or $\tra{\trmin}{s'_2}$. Let $p_i$ be the point which lies on $t_i(s_1)$ for $i=1,2,3$. By Lemma~\ref{vertex-intersection}, $t_3$ cannot intersect both $\tra{\trmin}{s_3}$ and $\tra{\trmin}{s'_2}$. Thus, $t_3$ intersects $t$ either through $\tra{\trmin}{s_3}$ or through $\tra{\trmin}{s'_2}$. We prove the former case; the proof for the latter case is similar. Assume that $t_3$ intersects $t$ through $\tra{\trmin}{s_3}$. By Lemma~\ref{triangle-intersection-lemma}, $t(p,p_3)\prec t_3$. See Figure~\ref{three-triangle-fig2}. In addition, both $t_1(s_3)$ and $t_2(s_3)$ are to the left of $t_3(s_3)$, because otherwise $q_3$ lies in $t_1\cup t_2\cup \hex{p}{q}$. If $q_3\in t_1\cup t_2$ we get a contradiction to Observation~\ref{no-point-in-triangle-obs}. If $q_3\in \hex{p}{q}$ then by Observation~\ref{obs2}, we have $t(p,q_3)\prec t$, and hence, the cycle $p,p_3,q_3,p$ contradicts Lemma~\ref{cycle-lemma}.

Without loss of generality assume that $\tra{t_1}{s_1}$ is above $\tra{t_2}{s_1}$; see Figure~\ref{three-triangle-fig2}. If $\tra{t_1}{v_3}$ is in $t_2$ or $\tra{t_2}{v_3}$ is in $t_1$, then we get a contradiction to Corollary~\ref{biniaz-cor}. Thus, assume that $\tra{t_1}{v_3}\notin t_2$ and $\tra{t_2}{v_3} \notin t_1$. This implies that either (i) $\tra{t_2}{s_3}$ is to the right of $\tra{t_1}{s_3}$ or (ii) $\tra{t_2}{s_2}$ is to the left of $\tra{t_1}{s_2}$. We show that both cases lead to a contradiction.

In case (i), $p_2$ lies in the interior of $\hex{p,q}$, and then by Observation~\ref{obs2}, we have $t(p,p_2)\prec t$; see Figure~\ref{three-triangle-fig2}(a). In addition, Lemma~\ref{triangle-intersection-lemma} implies that $t(p_2,q_3)\prec \max\{t,t_3\}\preceq t_3$. Thus, the cycle $p,p_3,q_3,p_2,p$ contradicts Lemma~\ref{cycle-lemma}.

Now consider case (ii) where $\tra{t_1}{s_1}$ is above $\tra{t_2}{s_1}$ and $\tra{t_2}{s_2}$ is to the left of $\tra{t_1}{s_2}$. If $p_1$ is to the right of $t$, then as in case (i), the cycle $p,p_3,q_3,p_1,p$ contradicts Lemma~\ref{cycle-lemma}. Thus, assume that $p_1$ is to the left of $t$, as shown in Figure~\ref{three-triangle-fig2}(b). By Lemma~\ref{triangle-intersection-lemma}, we have $t(q,p_1)\prec \max\{t,t_1\}\preceq t_1$. Each side of $t_1$ contains either $p_1$ or $q_1$, while $p_1$ is on the part of $t_1(s_1)$ which is to the left of $t$, thus, $q_1$ is on $\tra{t_1}{s_3}$. Consider the six cones around $\tra{t}{v_3}$; see Figure~\ref{three-triangle-fig2}(b). We have three cases: (a) $q_1\in C^5_{\tra{t}{v_3}}$, (b) $q_1\in C^6_{\tra{t}{v_3}}$ or (c) $q_1\in C^1_{\tra{t}{v_3}}$. 

In case (a), which is shown in Figure~\ref{three-triangle-fig2}(b), by Lemma~\ref{triangle3}, we have $\max\{t(p_1,p_2),t(q_1,q_2)\}\prec\max\{t_1,t_2\}$. Thus, the cycle $p_1,p_2,q_2,q_1,p_1$ contradicts Lemma~\ref{cycle-lemma}. In Case (b), which is shown in Figure~\ref{three-triangle-fig2}(c), we have $t(q_1,q_3)\prec t_3$, because if we map $t_3$ to a downward triangle $\tau$\textemdash of area equal to the area of $t_3$\textemdash which has $\tau(v_2)$ on $t(v_3)$, then $\tau$ contains both $q_1$ and $q_3$. Therefore, the cycle $p,p_3,q_3,q_1,p_1,q,p$ contradicts Lemma~\ref{cycle-lemma}. In Case (c), which is shown in Figure~\ref{three-triangle-fig2}(d), by Observation~\ref{obs2}, $t(p,q_1)\prec t$, and then, the cycle $p,q_1,p_1,q,p$ contradicts Lemma~\ref{cycle-lemma}.
\end{proof}

\begin{figure}[htb]
  \centering
\setlength{\tabcolsep}{0in}
  $\begin{tabular}{cc}
\multicolumn{1}{m{.5\columnwidth}}{\centering\includegraphics[width=.38\columnwidth]{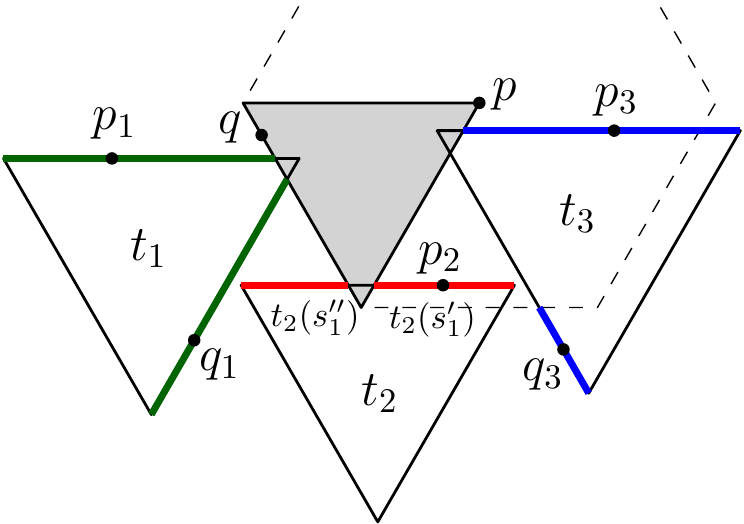}}
&\multicolumn{1}{m{.5\columnwidth}}{\centering\includegraphics[width=.38\columnwidth]{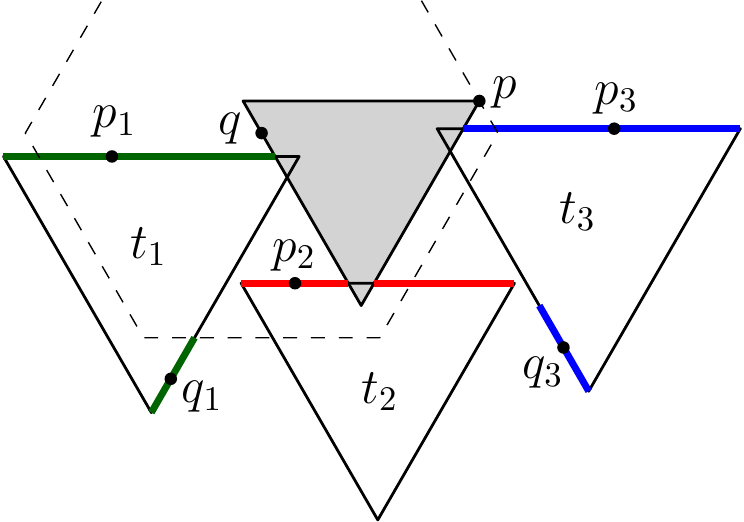}} \\
(a) & (b)
\end{tabular}$
  \caption{Illustration of Lemma~\ref{vertex-side-intersection-1}: (a) $p_2\in t_2(s'_1)$, and (b) $p_2\in t_2(s''_1)$.}
\label{three-triangle-fig}
\end{figure}

\begin{lemma}
\label{vertex-side-intersection-1}
If three triangles intersect $\trmin$ through $\tra{\trmin}{s'_2}, \tra{\trmin}{v_3}$ and $\tra{\trmin}{s_3}$. Then, at least one of the three triangles is not in $\mathcal{I}_2$. 
\end{lemma}
\begin{proof}
The proof is by contradiction. Assume that three triangles $t_1(p_1,q_1), t_2(p_2,q_2),t_3(p_3,q_3)$ in $\mathcal{I}_2$ intersect $\trmin$ through $\tra{\trmin}{s'_2}, \tra{\trmin}{v_3}, \tra{\trmin}{s_3}$, respectively. Let $p_i$ be the point which lies on $t_i(s_1)$ for $i=1,2,3$. See Figure~\ref{three-triangle-fig}(a). By Lemma~\ref{triangle-intersection-lemma}, we have $t(p,p_3)\prec t_3$ and $t(q,p_1)\prec t_1$. If $q_3$ is in the interior of $\hex{p}{q}$, then by Observation~\ref{obs2}, $t(p,q_3)\prec t$, and hence, the cycle $p,p_3,q_3,p$ contradicts Lemma~\ref{cycle-lemma}. If $q_1$ is in $\hex{q}{p}$, then by Observation~\ref{obs2}, $t(q,q_1)\prec t$, and hence, the cycle $q,q_1,p_1,q$ contradicts Lemma~\ref{cycle-lemma}; see Figure~\ref{three-triangle-fig}(b). Thus, assume that $q_3\notin \hex{p}{q}$ and $q_1\notin \hex{q}{p}$. Let $\tra{t_2}{s'_1}$ and $\tra{t_2}{s''_1}$ be the parts of $\tra{t_2}{s_1}$ which are to the right of $t(s_3)$ and to the left of $t(s_2)$, respectively. Consider the point $p_2$ which lies on $\tra{t_2}{s_1}$. 
If $p_2\in\tra{t_2}{s'_1}$, then $p_2\in \hex{p}{q}$ and by Observation~\ref{obs2}, $t(p,p_2)\prec t$. In addition, Lemma ~\ref{triangle-intersection-lemma} implies that $t(p_2,q_3)\prec t_3$. Thus, the cycle $p,p_3,q_3,p_2,p$ contradicts Lemma~\ref{cycle-lemma}; see Figure~\ref{three-triangle-fig}(a).
If $p_2\in\tra{t_2}{s''_1}$, then $p_2\in \hex{q}{p}$ and by Observation~\ref{obs2}, $t(q,p_2)\prec t$. In addition, Lemma ~\ref{triangle-intersection-lemma} implies that $t(p_2,q_1)\prec t_2$. Thus, the cycle $q,p_2,q_1,p_1,q$ contradicts Lemma~\ref{cycle-lemma}; see Figure~\ref{three-triangle-fig}(b).
\end{proof}

Putting Lemmas~\ref{side-intersection}, \ref{vertex-intersection}, \ref{vertex-side-intersection-2}, and \ref{vertex-side-intersection-1} together, implies that at most two triangles in $\mathcal{I}_2$ intersect $t$ through $t(s_3)\cup t(s'_2)$,  and consequently, at most two triangles in $\mathcal{I}_2$ intersect $t$ through $t(s_1)\cup t(s''_2)$. Thus, $\mathcal{I}_2$ contains at most four triangles. Recall that $\mathcal{I}_1$ contains at most four triangles. Then, $\mathcal{I}(e^+)$ has at most eight triangles. Therefore, the influence set of $e$, contains at most 9 edges (including $e$ itself). This completes the proof of Lemma~\ref{triangle-inf-lemma}. 

\begin{theorem}
\label{half-theta-six-thr}
Algorithm~\ref{alg1} computes a strong matching of size at least $\lceil\frac{n-1}{9}\rceil$ in $\G{\trids}{P}$.
\end{theorem}

\begin{figure}[htb]
  \centering
\includegraphics[width=.4\columnwidth]{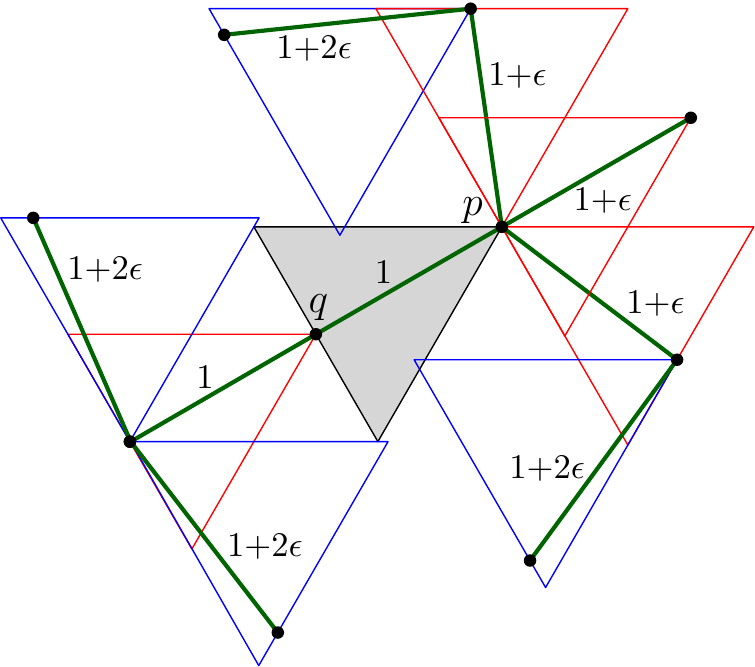}
  \caption{Four triangles in $\mathcal{I}_1$ (in red) and four triangles in $\mathcal{I}_2$ (in blue) intersect with $t(p,q)$.}
\label{five-fig}
\end{figure}

The bound obtained by Lemma~\ref{triangle-inf-lemma} is tight. Figure~\ref{five-fig} shows a configuration of 10 points in general position such that the influence set of a minimal edge is 9. In Figure~\ref{five-fig}, $t=t(p,q)$ represents a smallest edge of weight 1; the minimum spanning tree is shown in bold-green line segments. The weight of all edges\textemdash the area of the triangles representing these edges\textemdash is at least 1. The red triangles are in $\mathcal{I}_1$ and share either $p$ or $q$ with $t$. The blue triangles are in $\mathcal{I}_2$ and intersect $t$ through $t(s_1)\cup t(s''_2)$ or through $t(s_3)\cup t(s'_2)$; as show in Figure~\ref{five-fig}, two of them share only the points $t(v_2)$ and $t(v_3)$.

\section{Strong Matching in $G_{\bigtriangledown \hspace*{-8.7pt} \bigtriangleup}(P)$}

In this section we consider the problem of computing a strong matching in $\GUD(P)$. Recall that $\GUD(P)$ is the union of $\G{\trids}{P}$ and $\G{\trius}{P}$, and is equal to the graph $\Theta_6(P)$. We assume that $P$ is in general position, i.e., for each point $p\in P$, there is no point of $P\setminus \{p\}$ on $l_p^0$, $l_p^{60}$, and $l_p^{120}$. A matching $\mathcal{M}$ in $\GUD(P)$ is a strong matching if for each edge $e$ in $\mathcal{M}$ there is a homothet of $\trid$ or a homothet of $\triu$ representing $e$, such that these homothets are pairwise disjoint. See Figure~\ref{strong-example}(b). Using a similar approach as in~\cite{Abrego2009}, we prove the following theorem:

\label{theta-six-section}
\begin{theorem}
\label{theta-six-thr}
Let $P$ be a set of $n$ points in general position in the plane. Let $S$ be an upward or a downward equilateral-triangle that contains $P$. Then, it is possible to find a strong matching of size at least $\lceil\frac{n-1}{4}\rceil$ for $\GUD(P)$ in $S$.
\end{theorem}

\begin{proof}
The proof is by induction. Assume that any point set of size $n'\le n-1$ in a triangle $S'$, has a strong matching of size $\lceil \frac{n'-1}{4}\rceil$ in $S'$. Without loss of generality, assume $S$ is an upward equilateral-triangle. If $n$ is $0$ or $1$, then there is no matching in $S$, and if $n\in\{2, 3, 4, 5\}$, then by shrinking $S$, it is possible to find a strongly matched pair; the statement of the theorem holds. Suppose that $n\ge 6$, and $n=4m+r$, where $r\in\{0,1,2,3\}$. If $r\in\{0, 1,3\}$, then 
$\lceil \frac{n-1}{4}\rceil = \lceil \frac{(n-1)-1}{4}\rceil$, and by
induction we are done. Suppose that $n=4m+2$, for some $m\ge 1$. We prove that there are $\lceil\frac{n-1}{4}\rceil=m+1$ disjoint equilateral-triangles (upward or downward) in $S$,
each of them matches a pair of points in $P$. Partition $S$ into four equal area equilateral triangles $S_1, S_2, S_3, S_4$ containing $n_1, n_2,n_3, n_4$ points, respectively; see Figure~\ref{Theta-six-fig}(a). Let $n_i=4m_i+r_i$, where $r_i\in\{0,1,2,3\}$. 
By induction, in $S_1\cup S_2\cup S_3\cup S_4$, we have a strong matching of size at least
\begin{equation}
\label{A-eq}
A=\left\lceil\frac{n_1-1}{4}\right\rceil + \left\lceil\frac{n_2-1}{4}\right\rceil +\left\lceil\frac{n_3-1}{4}\right\rceil+\left\lceil \frac{n_4-1}{4}\right\rceil.
\end{equation}

{\bf Claim 1:} $A\ge m$.
\begin{proof}
By Equation~(\ref{A-eq}), we have
\begin{align*}
A&= \sum_{i=1}^{4}{\left\lceil\frac{n_i-1}{4}\right\rceil} \ge 
\sum_{i=1}^{4}\frac{n_i-1}{4} =\frac{n}{4} -1=\frac{4m+2}{4} -1=m-\frac{1}{2}.
\end{align*}
Since $A$ is an integer, we argue that $A\ge m$.
\end{proof}

If $A> m$, then we are done. Assume that $A=m$; in fact, by the induction hypothesis we have an strong matching of size $m$ for $P$. In order to complete the proof, we have to get one more strongly matched pair. Let $R$ be the multiset $\{r_1,r_2,r_3,r_4\}$.

{\bf Claim 2:} {\em If $A=m$, then either (i) one element in $R$ is equal to $3$ and the other elements are equal to $1$, or (ii) two elements in $R$ are equal to $0$ and the other elements are equal to $1$.}
\begin{proof}
Let $\alpha=r_1+r_2+r_3+r_4$, where $0\le r_i\le 3$. Then $n=4(m_1+m_2+m_3+m_4)+\alpha$. Since $n=4m+2$, $\alpha=4k+2$, for some $0\le k\le 2$. Thus, $n=4(m_1+m_2+m_3+m_4+k)+2$, where $m=m_1+m_2+m_3+m_4+k$.

By induction, in $S_i$, we get a matching of size at least $\lceil \frac{(4m_i+r_i)-1}{4}\rceil=m_i+\lceil \frac{r_i-1}{4}\rceil$. Hence, in $S_1\cup S_2\cup S_3\cup S_4$, we get a matching of size at least

\begin{equation*}
A=m_1+m_2+m_3+m_4+\left\lceil\frac{r_1-1}{4}\right\rceil + \left\lceil\frac{r_2-1}{4}\right\rceil +\left\lceil\frac{r_3-1}{4}\right\rceil+\left\lceil \frac{r_4-1}{4}\right\rceil.
\end{equation*}

Since $A=m$ and $m=m_1+m_2+m_3+m_4+k$, we have 

\begin{equation}
\label{k-eq}
k=\left\lceil\frac{r_1-1}{4}\right\rceil + \left\lceil\frac{r_2-1}{4}\right\rceil +\left\lceil\frac{r_3-1}{4}\right\rceil+\left\lceil \frac{r_4-1}{4}\right\rceil.
\end{equation}

Note that $0\le k\le 2$.
We go through some case analysis: (i) $k=0$, (ii) $k=1$, (iii) $k=2$. In case (i), we have $\alpha =4k+2=r_1+r_2+r_3+r_4=2$. In order to have $k$ equal to 0 in Equation~(\ref{k-eq}), no element in $R$ should be more than 1; this happens only if two elements in $R$ are equal to 0 and the other two elements are equal to 1. In case (ii), we have $\alpha =r_1+r_2+r_3+r_4=6$. In order to have $k$ equal to 1 in Equation~(\ref{k-eq}), at most one element in $R$ should be greater than 1; this happens only if three elements in $R$ are equal to 1 and the other element is equal to 3 (note that all elements in $R$ are smaller than 4). In case (iii), we have $\alpha =r_1+r_2+r_3+r_4=10$. In order to have $k$ equal to 2 in Equation~(\ref{k-eq}), at most two elements in $R$ should be greater than 1; which is not possible.
\end{proof} 
We show how to find one more matched pair in each case of Claim 2.

We define $\SM{$x$}{1}$ as the smallest upward equilateral-triangle contained in $S_1$ and anchored at the top corner of $S_1$, which contains all the points in $S_1$ except $x$ points. If $S_1$ contains less than $x$ points, then the area of $\SM{$x$}{1}$ is zero. We also define $\SP{$x$}{1}$ as the smallest upward equilateral-triangle that contains $S_1$ and anchored at the top corner of $S_1$, which has all the points in $S_1$ plus $x$ other points of $P$. Similarly we define upward triangles $\SM{$x$}{2}$ and $\SP{$x$}{2}$ which are anchored at the left corner of $S_2$. Moreover, we define upward triangles $\SM{$x$}{4}$ and $\SP{$x$}{4}$ which are anchored at the right corner of $S_4$. We define downward triangles $\SM{$x$}{3l}$, $\SM{$x$}{3r}$, $\SM{$x$}{3b}$ which are anchored at the top-left corner, top-right corner, and bottom corner of $S_3$, respectively. See Figure~\ref{Theta-six-fig}(a). 

{\bf Case 1:} {\em One element in $R$ is equal to 3 and the other elements are equal to 1.}

In this case, we have $m=m_1+m_2+m_3+m_4+1$. Because of the symmetry, we have two cases: (i) $r_3=3$, (ii) $r_j=3$ for some $j\in\{1,2,4\}$.

\begin{itemize}
 
\begin{figure}[h!]
  \centering
\setlength{\tabcolsep}{0in}
  $\begin{tabular}{cc}
\multicolumn{1}{m{.5\columnwidth}}{\centering\includegraphics[width=.38\columnwidth]{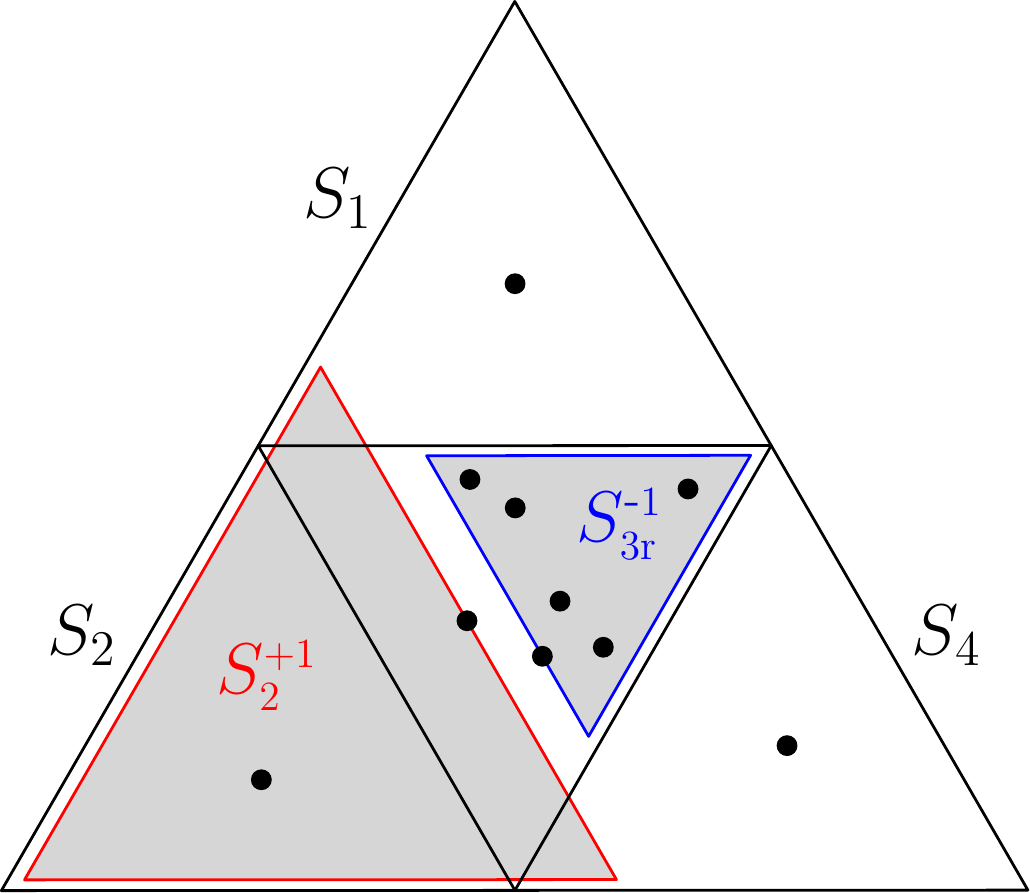}}
&\multicolumn{1}{m{.5\columnwidth}}{\centering\includegraphics[width=.38\columnwidth]{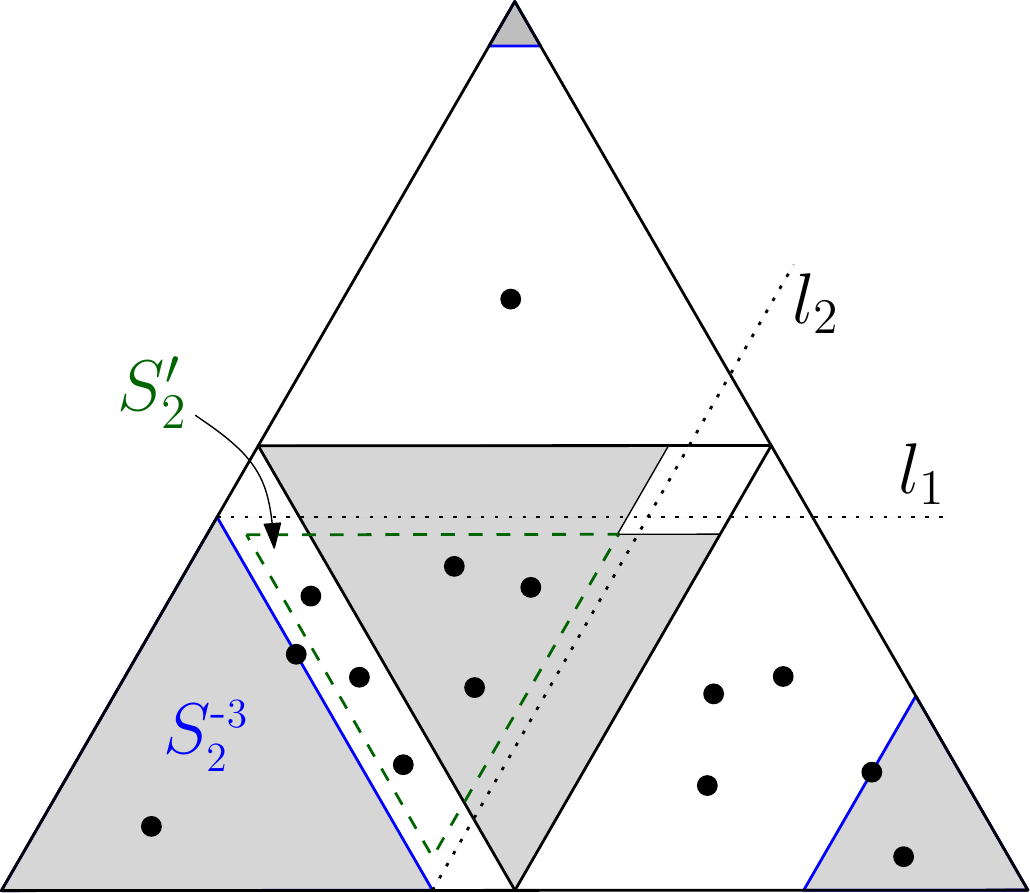}} \\
(a) & (b)
\end{tabular}$
  \caption{(a) Split $S$ into four equal area triangles. (b) $\SM{3}{2}$ is larger than $\SM{3}{1}$ and $\SM{3}{4}$.}
\label{Theta-six-fig}
\end{figure}

 \item {$r_3=3$.}

In this case $n_3=4m_3+3$. We differentiate between two cases, where all the elements of the multiset $\{m_1, m_2, m_4\}$ are equal to zero, or some of them are greater than zero.

\begin{itemize}
 \item {\em All elements of $\{m_1, m_2, m_4\}$ are equal zero.} In this case, we have $m=m_3+1$. Consider the triangles $\SP{1}{2}$ and $\SM{1}{3r}$. See Figure~\ref{Theta-six-fig}(a). Note that $\SP{1}{2}$ and $\SM{1}{3r}$ are disjoint, $\SP{1}{2}$ contains two points, and $\SM{1}{3r}$ contains $4m_3+2$ points. By induction, we get a matched pair in $\SP{1}{2}$ and a matching of size at least $m_3+1$ in $\SM{1}{3r}$. Thus, in total, we get a matching of size at least $1+(m_3+1)=m+1$ in $S$.

 \item {\em Some elements of $\{m_1, m_2, m_4\}$ are greater than zero.} Consider the triangles $\SM{3}{1}$, $\SM{3}{2}$, and $\SM{3}{4}$. Note that the area of some of these triangles\textemdash but not all\textemdash may be equal to zero. See Figure~\ref{Theta-six-fig}(b). By induction, we get matchings of size $m_1$, $m_2$, and $m_4$ in $\SM{3}{1}$, $\SM{3}{2}$, and $\SM{3}{4}$, respectively. Without loss of generality, assume $\SM{3}{2}$, is larger than $\SM{3}{1}$ and $\SM{3}{4}$. Consider the half-lines $l_1$ and $l_2$ which are parallel to $l^0$ and $l^{60}$ axis, and have their endpoints on the top corner and right corner of $\SM{3}{2}$, respectively. We define $S'_2$ as the downward equilateral-triangle which is bounded by $l_1$, $l_2$, and the right side of $\SM{3}{2}$; the dashed triangle in Figure~\ref{Theta-six-fig}(b). Note that $l_1$ and $l_2$ do not intersect $\SM{3}{1}$ and $\SM{3}{4}$. In addition, $\SM{3}{1}$, $\SM{3}{2}$, $\SM{3}{4}$, and $S'_2$ are pairwise disjoint. If any point of $S_1\cup S_2\cup S_3$ is to the right of $l_2$, then consider $\SP{1}{4}$ and $\SM{1}{3l}$. By induction, we get a matching of size $m_1+m_2+(m_3+1)+(m_4+1)$ in $\SM{3}{1}\cup \SM{3}{2}\cup\SM{1}{3l}\cup \SP{1}{4}$, and hence a matching of size $m+1$ in $S$. If any point of $S_2\cup S_3\cup S_4$ is above $l_1$, then consider $\SP{1}{1}$ and $\SM{1}{3b}$. By induction, we get a matching of size $(m_1+1)+m_2+(m_3+1)+ m_4$ in $\SP{1}{1}\cup \SM{3}{2}\cup\SM{1}{3b}\cup \SM{3}{4}$, and hence a matching of size $m+1$ in $S$. Otherwise, $S'_2$ contains $n_3+3=4(m_3+1)+2$ points. Thus, by induction, we get a matching of size $m_1+m_2+(m_3+2)+ m_4$ in $S_1\cup \SM{3}{2}\cup S'_2\cup S_4$, and hence a matching of size $m+1$ in $S$.
\end{itemize}

\item {\em $r_j=3$, for some $j\in\{1,2,4\}$.}

Without loss of generality, assume that $r_j=r_2$. Then, $n_2=4m_2+3$. Consider the triangles $\SM{3}{1}$, $\SM{1}{2}$, and $\SM{3}{4}$. See Figure~\ref{Theta-six-fig2}(a). By induction, we get matchings of size $m_1$, $m_2+1$, and $m_4$ in $\SM{3}{1}$, $\SM{1}{2}$, and $\SM{3}{4}$, respectively. 
Now we consider the largest triangle among $\SM{3}{1}$, $\SM{1}{2}$, and $\SM{3}{4}$. Because of the symmetry, we have two cases: (i) $\SM{1}{2}$ is the largest, or (ii) $\SM{3}{4}$ is the largest.
\begin{figure}[h!]
  \centering
\setlength{\tabcolsep}{0in}
  $\begin{tabular}{cc}
\multicolumn{1}{m{.5\columnwidth}}{\centering\includegraphics[width=.38\columnwidth]{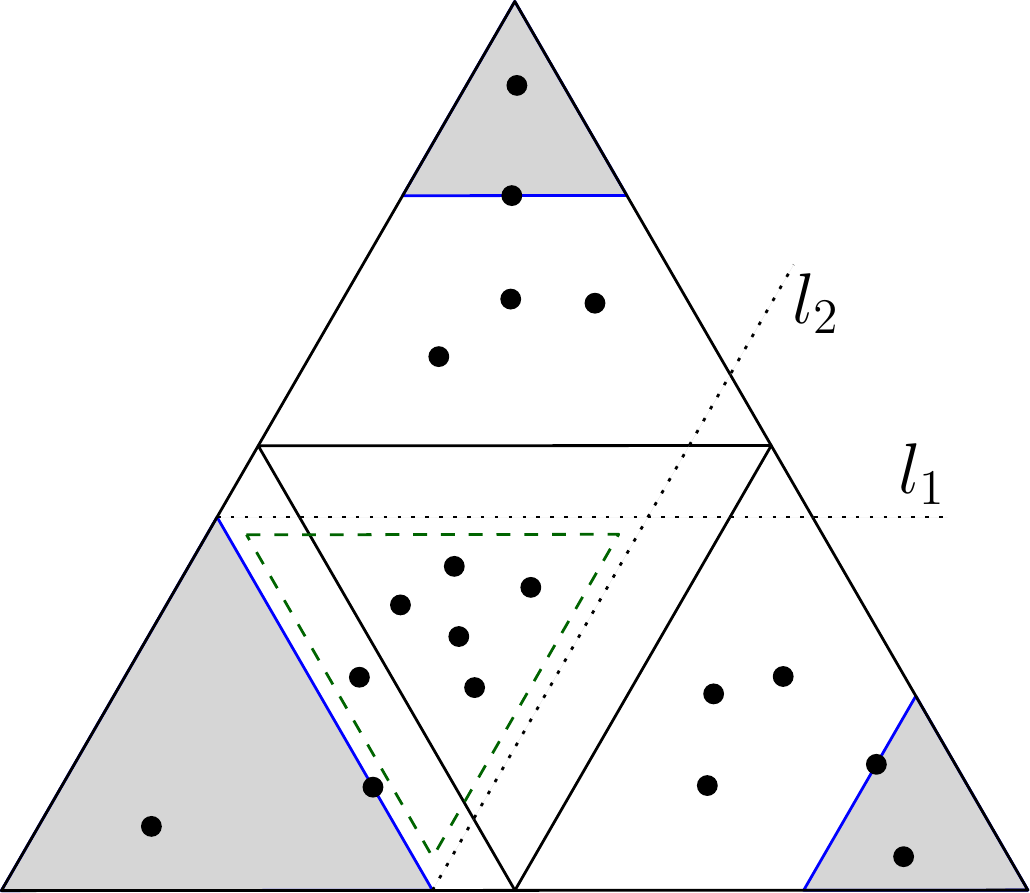}}
&\multicolumn{1}{m{.5\columnwidth}}{\centering\includegraphics[width=.38\columnwidth]{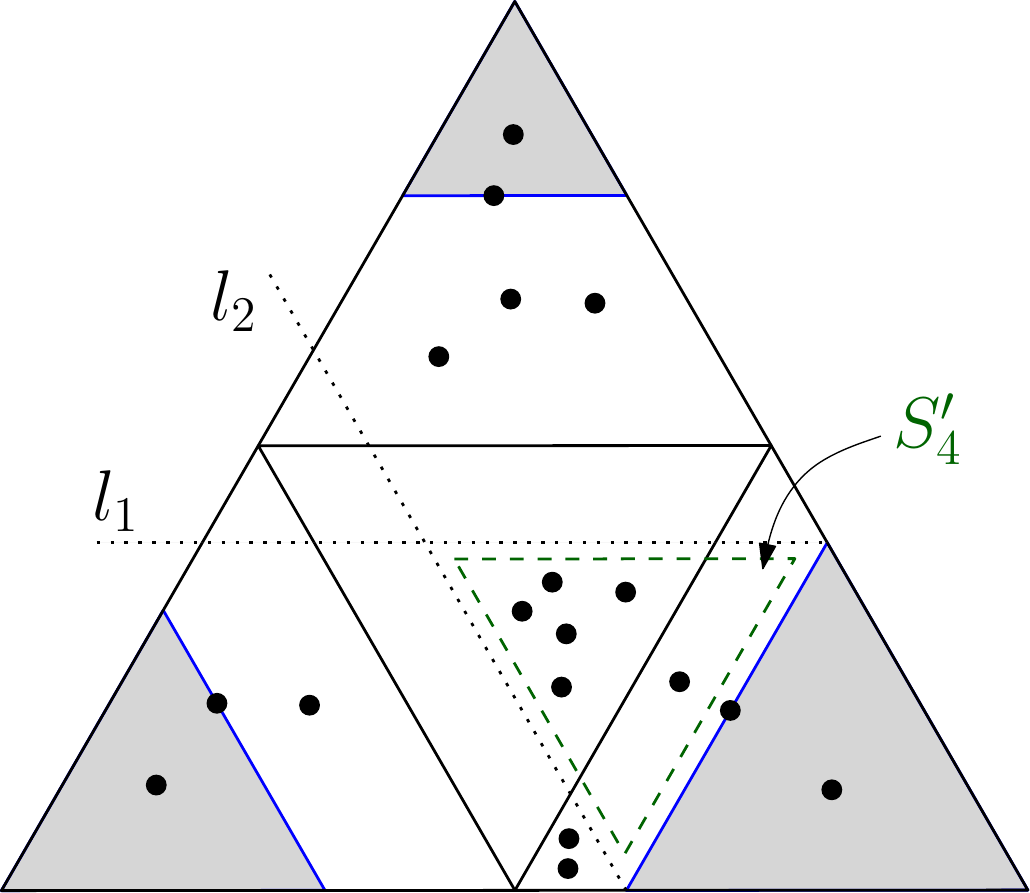}}
\\
(a) & (b)
\end{tabular}$
  \caption{(a) $\SM{1}{2}$ is larger than $\SM{3}{1}$ and $\SM{3}{4}$. (b) $\SM{3}{4}$ is larger than $\SM{3}{1}$ and $\SM{1}{2}$.}
\label{Theta-six-fig2}
\end{figure}
\begin{itemize}
 \item {\em $\SM{1}{2}$ is larger than $\SM{3}{1}$ and $\SM{3}{4}$.}
Define the half-lines $l_1$, $l_2$, and the triangle $S'_2$ as in the previous case. See Figure~\ref{Theta-six-fig2}(a). If any point of $S_1\cup S_2\cup S_3$ is to the right of $l_2$, then consider $\SP{1}{4}$ and $\SM{1}{3l}$. By induction, we get a matching of size $m_1+(m_2+1)+m_3+(m_4+1)$ in $\SM{3}{1}\cup \SM{1}{2}\cup\SM{1}{3l}\cup \SP{1}{4}$. If any point of $S_2\cup S_3\cup S_4$ is above $l_1$, then consider $\SP{1}{1}$ and $\SM{1}{3b}$. By induction, we get a matching of size $(m_1+1)+(m_2+1)+m_3+m_4$ in $\SP{1}{1}\cup \SM{1}{2}\cup\SM{1}{3b}\cup \SM{3}{4}$. Otherwise, $S'_2$ contains $n_3+1=4m_3+2$ points. Thus, by induction, we get a matching of size $m_1+(m_2+1)+(m_3+1)+ m_4$ in $S_1\cup \SM{1}{2}\cup S'_2\cup S_4$. As a result, in all cases we get a matching of size $m+1$ in $S$.

\item {\em $\SM{3}{4}$ is larger than $\SM{3}{1}$ and $\SM{1}{2}$.}
Define the half-lines $l_1$, $l_2$, and the triangle $S'_4$ as in Figure~\ref{Theta-six-fig2}(b). If any point of $S_1\cup S_3\cup S_4$ is above $l_1$, then by induction, we get a matching of size $(m_1+1)+(m_2+1)+m_3+m_4$ in $\SP{1}{1}\cup \SM{1}{2}\cup\SM{1}{3b}\cup \SP{3}{4}$. If at least three points of $S_1\cup S_3\cup S_4$ are to the left of $l_2$, then consider $\SP{3}{2}$ and $\SM{3}{3r}$. Note that $\SP{3}{2}$ contains $n_2+3=4(m_2+1)+2$ points. By induction, we get a matching of size $m_1+(m_2+2)+m_3+m_4$ in $\SM{3}{1}\cup \SP{3}{2}\cup\SM{3}{3r}\cup \SM{3}{4}$. Otherwise, $S'_4$ contains at least $n_3+1=4m_3+2$ points. Thus, by induction, we get a matching of size $m_1+(m_2+1)+(m_3+1)+ m_4$ in $S_1\cup S_2\cup S'_4\cup \SM{3}{4}$. As a result, in all cases we get a matching of size $m+1$ in $S$.
\end{itemize}
\end{itemize}

{\bf Case 2:} {\em Two elements in $R$ are equal to 0 and the other elements are equal to 1.}

In this case, we have $m=m_1+m_2+m_3+m_4$. Again, because of the symmetry, we have two cases: (i) $r_3=0$, (ii) $r_3\neq 0$.

\begin{itemize}
 \item $r_3=0.$

Without loss of generality assume that $r_2=0$ and $r_1=r_4=1$. Thus, $n_1=4m_1+1$, $n_2=4m_2$, $n_3=4m_3$, and $n_4=4m_4+1$. If all elements of $\{m_1,m_2,m_4\}$ are equal to zero, then we have $m=m_3$, where $m_3\ge 1$. Consider the triangles $\SP{1}{4}$ and $\SM{1}{3l}$, which are disjoint. By induction, we get a matched pair in $\SP{1}{4}$ and a matching of size at least $m_3$ in $\SM{1}{3l}$. Thus, in total, we get a matching of size at least $1+m_3=m+1$ in $S$. Assume some elements in $\{m_1,m_2,m_4\}$ are greater than zero. Consider the triangles $\SM{3}{1}$, $\SM{2}{2}$, and $\SM{3}{4}$. See Figure~\ref{Theta-six-fig3}(a). By induction, we get a matching of size $m_1$, $m_2$, and $m_4$ in $\SM{3}{1}$, $\SM{2}{2}$, and $\SM{3}{4}$, respectively. 
Now we consider the largest triangle among $\SM{3}{1}$, $\SM{2}{2}$, and $\SM{3}{4}$. Because of the symmetry, we have two cases: (i) $\SM{2}{2}$ is the largest, or (ii) $\SM{3}{4}$ is the largest.

\begin{figure}[h!]
  \centering
\setlength{\tabcolsep}{0in}
  $\begin{tabular}{cc}
\multicolumn{1}{m{.5\columnwidth}}{\centering\includegraphics[width=.38\columnwidth]{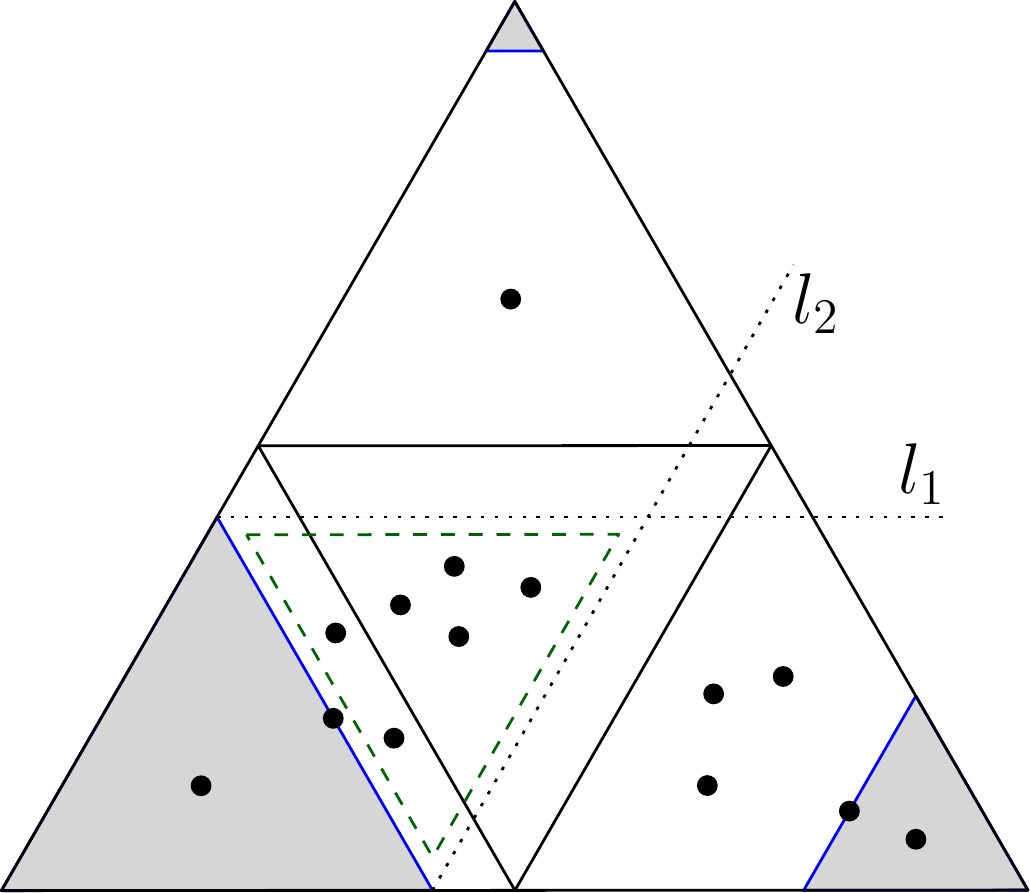}}
&\multicolumn{1}{m{.5\columnwidth}}{\centering\includegraphics[width=.38\columnwidth]{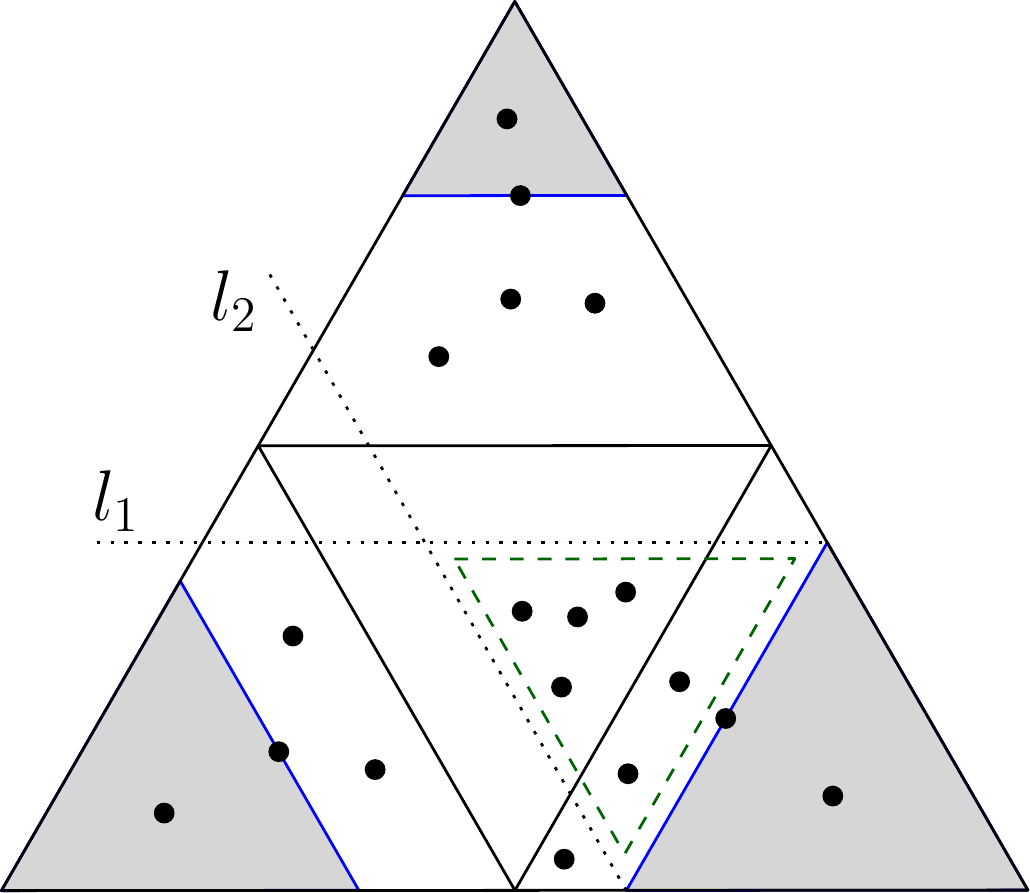}}
\\
(a) & (b)
\end{tabular}$
  \caption{(a) $\SM{2}{2}$ is larger than $\SM{3}{1}$ and $\SM{3}{4}$. (b) $\SM{3}{4}$ is larger than $\SM{3}{1}$ and $\SM{2}{2}$.}
\label{Theta-six-fig3}
\end{figure}
\begin{itemize}
 \item {\em $\SM{2}{2}$ is larger than $\SM{3}{1}$ and $\SM{3}{4}$.}
Define $l_1$, $l_2$, $S'_2$ as in Figure~\ref{Theta-six-fig3}(a). If any point of $S_1\cup S_2\cup S_3$ is to the right of $l_2$, then by induction, we get a matching of size $m_1+m_2+m_3+(m_4+1)$ in $\SM{3}{1}\cup \SM{2}{2}\cup\SM{1}{3l}\cup \SP{1}{4}$. If any point of $S_2\cup S_3\cup S_4$ is above $l_1$, then by induction, we get a matching of size $(m_1+1)+m_2+m_3+m_4$ in $\SP{1}{1}\cup \SM{2}{2}\cup\SM{1}{3b}\cup \SM{3}{4}$. Otherwise, $S'_2$ contains $n_3+2=4m_3+2$ points. Thus, by induction, we get a matching of size $m_1+m_2+(m_3+1)+ m_4$ in $S_1\cup \SM{2}{2}\cup S'_2\cup S_4$. In all cases we get a matching of size $m+1$ in $S$.

\item {\em $\SM{3}{4}$ is larger than $\SM{3}{1}$ and $\SM{2}{2}$.}
Define $l_1$, $l_2$, $S'_4$ as in Figure~\ref{Theta-six-fig3}(b). If any point of $S_1\cup S_3\cup S_4$ is above $l_1$, then by induction, we get a matching of size $(m_1+1)+m_2+m_3+m_4$ in $\SP{1}{1}\cup \SM{2}{2}\cup\SM{1}{3b}\cup \SP{3}{4}$. If at least two points of $S_1\cup S_3\cup S_4$ are to the left of $l_2$, then by induction, we get a matching of size $m_1+(m_2+1)+m_3+m_4$ in $\SM{3}{1}\cup \SP{2}{2}\cup\SM{2}{3r}\cup \SM{3}{4}$. Otherwise, $S'_4$ contains at least $n_3+2=4m_3+2$ points. Thus, by induction, we get a matching of size $m_1+m_2+(m_3+1)+ m_4$ in $S_1\cup S_2\cup S'_4\cup \SM{3}{4}$. In all cases we get a matching of size $m+1$ in $S$.
\end{itemize}
  \item $r_3\neq 0.$

In this case $r_3=1$, and without loss of generality, assume that $r_2=1$; which means $r_1=r_4=0$. Thus, $n_1=4m_1$, $n_2=4m_2+1$, $n_3=4m_3+1$, and $n_4=4m_4$. If all elements of $\{m_1,m_2,m_4\}$ are equal to zero, then we have $m=m_3$, where $m_3\ge 1$. Consider the triangles $\SP{1}{2}$ and $\SM{1}{3r}$, which are disjoint. By induction, we get a matched pair in $\SP{1}{2}$ and a matching of size at least $m_3$ in $\SM{1}{3r}$. Thus, in total, we get a matching of size at least $1+m_3=m+1$ in $S$. Assume some elements in $\{m_1,m_2,m_4\}$ are greater than zero. Consider the triangles $\SM{2}{1}$, $\SM{3}{2}$, and $\SM{2}{4}$. See Figure~\ref{Theta-six-fig4}(a). By induction, we get matchings of size $m_1$, $m_2$, and $m_4$ in $\SM{2}{1}$, $\SM{3}{2}$, and $\SM{2}{4}$, respectively. 
Now we consider the largest triangle among $\SM{2}{1}$, $\SM{3}{2}$, and $\SM{2}{4}$. Because of symmetry, we have two cases: (i) $\SM{3}{2}$ is the largest, or (ii) $\SM{2}{4}$ is the largest.

\begin{figure}[h!]
  \centering
\setlength{\tabcolsep}{0in}
  $\begin{tabular}{cc}
\multicolumn{1}{m{.5\columnwidth}}{\centering\includegraphics[width=.38\columnwidth]{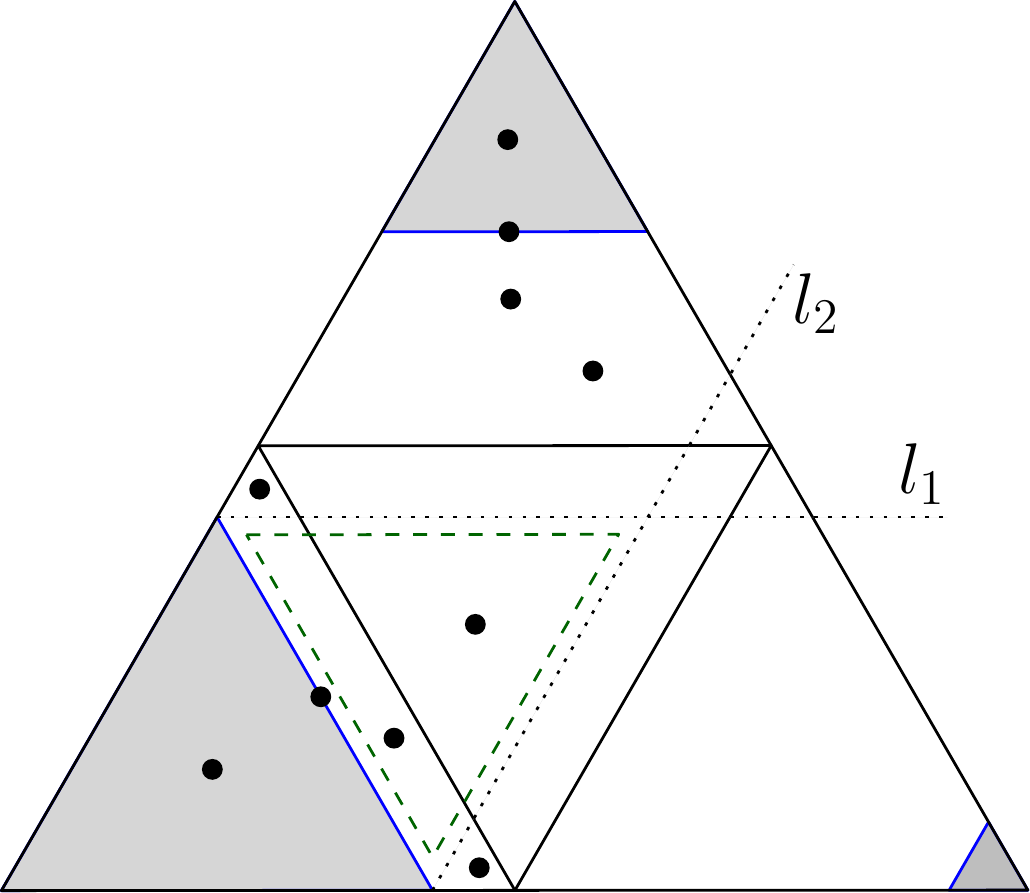}}
&\multicolumn{1}{m{.5\columnwidth}}{\centering\includegraphics[width=.38\columnwidth]{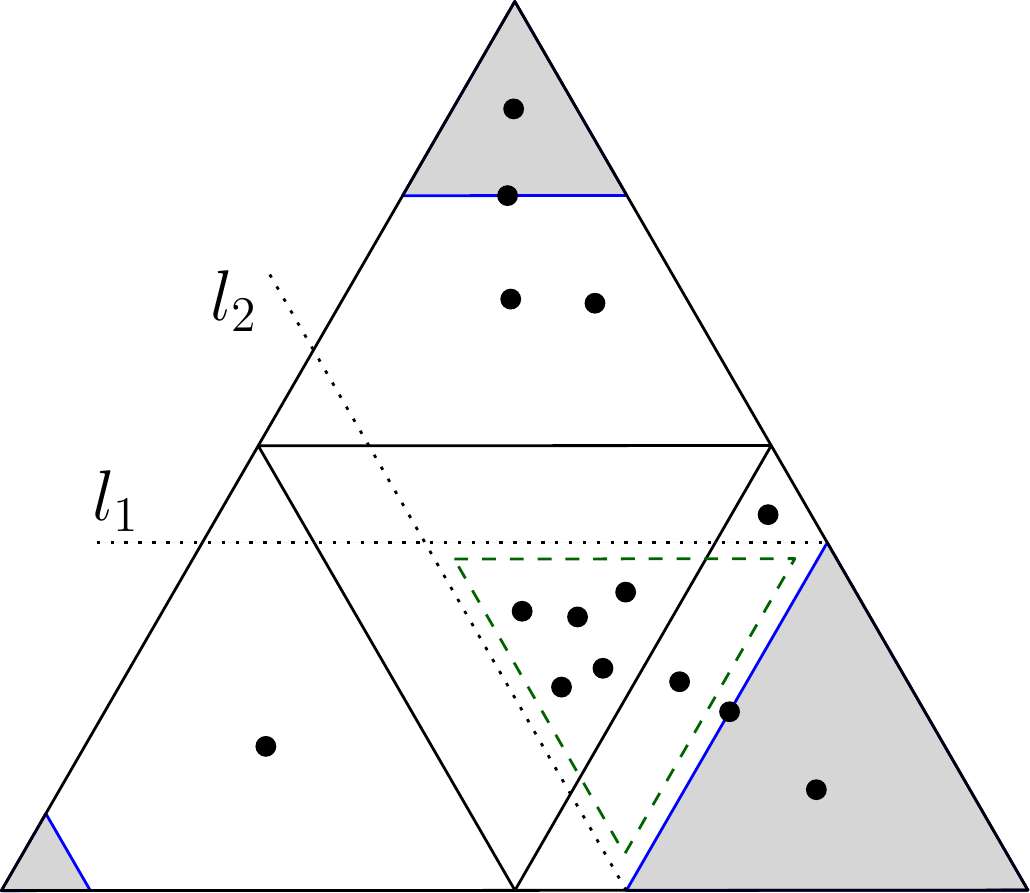}}
\\
(a) & (b)
\end{tabular}$
  \caption{(a) $\SM{3}{2}$ is larger than $\SM{2}{1}$ and $\SM{2}{4}$. (b) $\SM{2}{4}$ is larger than $\SM{2}{1}$ and $\SM{3}{2}$.}
\label{Theta-six-fig4}
\end{figure}
\begin{itemize}
 \item {\em $\SM{3}{2}$ is larger than $\SM{2}{1}$ and $\SM{2}{4}$.}
Define $l_1$, $l_2$, $S'_2$ as in Figure~\ref{Theta-six-fig4}(a). If at least two points of $S_1\cup S_2\cup S_3$ are to the right of $l_2$, then by induction, we get a matching of size $m_1+m_2+m_3+(m_4+1)$ in $\SM{2}{1}\cup \SM{3}{2}\cup\SM{2}{3l}\cup \SP{2}{4}$. If at least two points of $S_2\cup S_3\cup S_4$ are above $l_1$, then by induction, we get a matching of size $(m_1+1)+m_2+m_3+m_4$ in $\SP{2}{1}\cup \SM{3}{2}\cup\SM{2}{3b}\cup \SM{2}{4}$. Otherwise, $S'_2$ contains $n_3+1=4m_3+2$ points, and we get a matching of size $m_1+m_2+(m_3+1)+ m_4$ in $S_1\cup \SM{3}{2}\cup S'_2\cup S_4$. In all cases we get a matching of size $m+1$ in $S$.

\item {\em $\SM{2}{4}$ is larger than $\SM{2}{1}$ and $\SM{3}{2}$.}
Define $l_1$, $l_2$, $S'_4$ as in Figure~\ref{Theta-six-fig4}(b). If at least two points of $S_2\cup S_3\cup S_4$ are above $l_1$, then by induction, we get a matching of size $(m_1+1)+m_2+m_3+m_4$ in $\SP{2}{1}\cup \SM{3}{2}\cup\SM{2}{3b}\cup \SM{2}{4}$. If any point of $S_1\cup S_3\cup S_4$ is to the left of $l_2$, then by induction, we get a matching of size $m_1+(m_2+1)+m_3+m_4$ in $\SM{2}{1}\cup \SP{1}{2}\cup\SM{1}{3r}\cup \SM{2}{4}$. Otherwise, $S'_4$ contains at least $n_3+1=4m_3+2$ points, and we get a matching of size $m_1+m_2+(m_3+1)+ m_4$ in $S_1\cup S_2\cup S'_4\cup \SM{2}{4}$. In all cases we get a matching of size $m+1$ in $S$. 
\end{itemize}
\end{itemize}
\end{proof}
\section{Strong Matching in 
$\G{\sqrs}{P}$}

In this section we consider the problem of computing a strong matching in $\G{\sqrs}{P}$, where $\sqr$ is an axis-aligned square whose center is the origin. We assume that $P$ is in general position, i.e., (i) no two points have the same $x$-coordinate or the same $y$-coordinate, and (ii) no four points are on the boundary of any homothet of $\sqr$. Recall that $\G{\sqrs}{P}$ is equal to the $L_\infty$-Delaunay graph on $P$. \'{A}brego et al. \cite{Abrego2004, Abrego2009} proved that $\G{\sqrs}{P}$ has a strong matching of size at least $\lceil n/5\rceil$. Using a similar approach as in Section~\ref{theta-six-section}, we prove that $\G{\sqrs}{P}$ has a strong  matching of size at least $\lceil\frac{n-1}{4}\rceil$.

\label{infty-Delaunay-section}
\begin{theorem}
\label{infty-Delaunay-thr}
Let $P$ be a set of $n$ points in general position in the plane. Let $S$ be an axis-parallel square that contains $P$. Then, it is possible to find a strong matching of size at least $\lceil\frac{n-1}{4}\rceil$ for $\G{\sqrs}{P}$ in $S$.
\end{theorem}

\begin{proof}
The proof is by induction. Assume that any point set of size $n'\le n-1$ in an axis-parallel square $S'$, has a strong matching of size $\lceil \frac{n'-1}{4}\rceil$ in $S'$. If $n$ is $0$ or $1$, then there is no matching in $S$, and if $n\in\{2, 3, 4, 5\}$, then by shrinking $S$, it is possible to find a strongly matched pair. Suppose that $n\ge 6$, and $n=4m+r$, where $r\in\{0,1,2,3\}$. If $r\in\{0, 1,3\}$, then 
$\lceil \frac{n-1}{4}\rceil = \lceil \frac{(n-1)-1}{4}\rceil$, and by
induction we are done. Suppose that $n=4m+2$, for some $m\ge 1$. We prove that there are $\lceil\frac{n-1}{4}\rceil=m+1$ disjoint squares in $S$,
each of them matches a pair of points in $P$. Partition $S$ into four equal area squares $S_1, S_2, S_3, S_4$ which contain $n_1, n_2,n_3, n_4$ points, respectively; see Figure~\ref{square-fig1}(a). Let $n_i=4m_i+r_i$ for $1\le i\le 4$, where $r_i\in\{0,1,2,3\}$. Let $R$ be the multiset $\{r_1,r_2,r_3,r_4\}$. 
By induction, in $S_1\cup S_2\cup S_3\cup S_4$, we have a strong matching of size at least
\begin{equation}
A=\left\lceil\frac{n_1-1}{4}\right\rceil + \left\lceil\frac{n_2-1}{4}\right\rceil +\left\lceil\frac{n_3-1}{4}\right\rceil+\left\lceil \frac{n_4-1}{4}\right\rceil.\nonumber
\end{equation} 

In the proof of Theorem~\ref{theta-six-thr}, we have shown the following two claims:

{\bf Claim 1:} {$A\ge m$.}

{\bf Claim 2:} {\em If $A=m$, then either (i) one element in $R$ is equal to $3$ and the other elements are equal to $1$, or (ii) two elements in $R$ are equal to $0$ and the other elements are equal to $1$.}

If $A> m$, then we are done. Assume that $A=m$; in fact, by the induction hypothesis we have an strong matching of size $m$ in $S$. 
We show how to find one more strongly matched pair in each case of Claim 2.

We define $\SM{$x$}{1}$ as the smallest axis-parallel square contained in $S_1$ and anchored at the top-left corner of $S_1$, which contains all the points in $S_1$ except $x$ points. If $S_1$ contains less than $x$ points, then the area of $\SM{$x$}{1}$ is zero. We also define $\SP{$x$}{1}$ as the smallest axis-parallel square that contains $S_1$ and anchored at the top-left corner of $S_1$, which has all the points in $S_1$ plus $x$ other points of $P$. See Figure~\ref{square-fig1}(a). Similarly we define the squares $\SM{$x$}{2}$, $\SP{$x$}{2}$, and $\SM{$x$}{3}$, $\SP{$x$}{3}$, and $\SM{$x$}{4}$, $\SP{$x$}{4}$ which are anchored at the top-right corner of $S_2$, and the bottom-left corner of $S_3$, and the bottom-right corner of $S_4$, respectively.

{\bf Case 1:} {\em One element in $R$ is equal to 3 and the other elements are equal to 1.}

In this case, we have $m = m_1 + m_2 + m_3 + m_4 + 1$. Without loss of generality, assume that $r_1=3$ and $r_2=r_3=r_4=1$. Consider the squares $\SM{1}{1}$, $\SM{3}{2}$, $\SM{3}{3}$, and $\SM{3}{4}$. Note that the area of some of these squares\textemdash but not all\textemdash may be
equal to zero. See Figure~\ref{square-fig1}(b). By induction, we get matchings of size $m_1+1$, $m_2$, $m_3$, and $m_4$, in 
$\SM{1}{1}$, $\SM{3}{2}$, $\SM{3}{3}$, and $\SM{3}{4}$, respectively. Now consider the largest square among $\SM{1}{1}$, $\SM{3}{2}$, $\SM{3}{3}$, and $\SM{3}{4}$. Because of the symmetry, we have only three cases: (i) $\SM{1}{1}$ is the largest, (ii) $\SM{3}{2}$ is the largest, and (iii) $\SM{3}{4}$ is the largest.
\begin{figure}[htb]
  \centering
\setlength{\tabcolsep}{0in}
  $\begin{tabular}{ccc}
\multicolumn{1}{m{.33\columnwidth}}{\centering\includegraphics[width=.26\columnwidth]{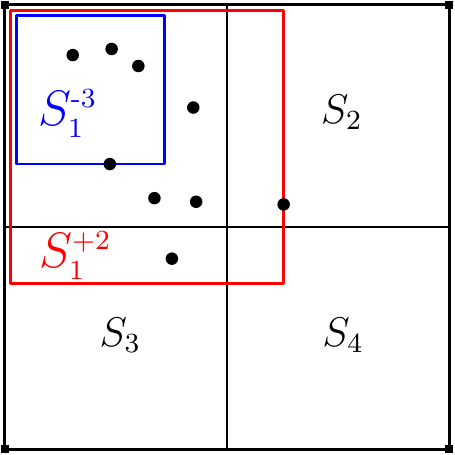}}
&\multicolumn{1}{m{.33\columnwidth}}{\centering\includegraphics[width=.26\columnwidth]{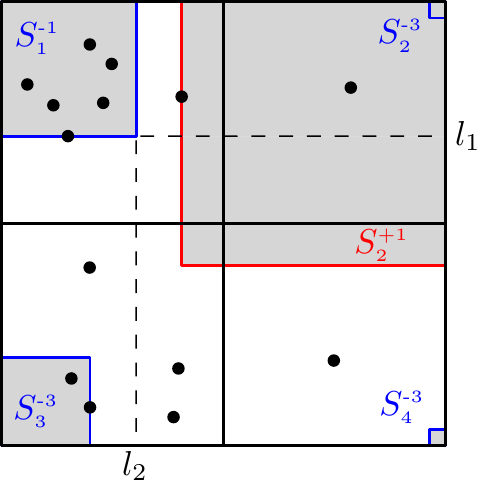}} &\multicolumn{1}{m{.33\columnwidth}}{\centering\includegraphics[width=.26\columnwidth]{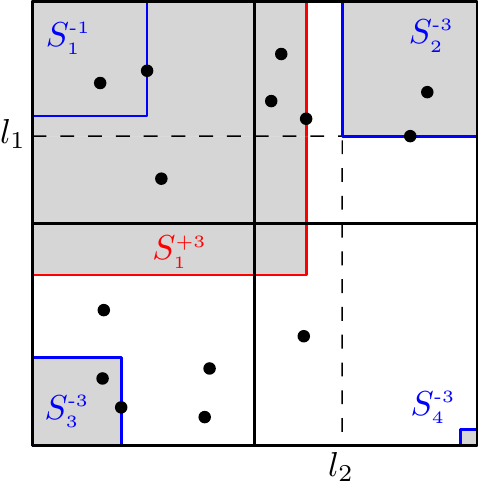}}
\\
(a) & (b)& (c)
\end{tabular}$
  \caption{(a) Split $S$ into four equal area squares. (b) $\SM{1}{1}$ is larger than $\SM{3}{2}$, $\SM{3}{3}$, and $\SM{3}{4}$. (c) $\SM{3}{2}$ is larger than $\SM{1}{1}$, $\SM{3}{3}$, and $\SM{3}{4}$.}
\label{square-fig1}
\end{figure}
\begin{itemize}
\item {\em $\SM{1}{1}$ is the largest square.}
Consider the lines $l_1$ and $l_2$ which contain the bottom side and right side of $\SM{1}{1}$, respectively; the dashed lines in Figure~\ref{square-fig1}(b). Note that $l_1$ and $l_2$ do not intersect any of $\SM{3}{2}$, $\SM{3}{3}$, and $\SM{3}{4}$. If any point of $S_1$ is to the right of $l_2$, then by induction, we get a matching of size $(m_1+1)+(m_2+1)+m_3+m_4$ in $\SM{1}{1}\cup \SP{1}{2}\cup\SM{3}{3}\cup \SM{3}{4}$. Otherwise, by induction, we get a matching of size $(m_1+1)+m_2+(m_3+1)+ m_4$ in $\SM{1}{1}\cup \SM{3}{2}\cup\SP{1}{3}\cup \SM{3}{4}$. In all cases we get a matching of size $m+1$ in $S$. 
\item {\em $\SM{3}{2}$ is the largest square.}
Consider the lines $l_1$ and $l_2$ which contain the bottom side and left side of $\SM{3}{2}$, respectively; the dashed lines in Figure~\ref{square-fig1}(c). Note that $l_1$ and $l_2$ do not intersect any of $\SM{1}{1}$, $\SM{3}{3}$, and $\SM{3}{4}$. If any point of $S_2$ is below $l_1$, then by induction, we get a matching of size $(m_1+1)+m_2+m_3+(m_4+1)$ in $\SM{1}{1}\cup \SM{3}{2}\cup\SM{3}{3}\cup \SP{1}{4}$. Otherwise, by induction, we get a matching of size $(m_1+2)+m_2+m_3+ m_4$ in $\SP{3}{1}\cup \SM{3}{2}\cup\SM{3}{3}\cup \SM{3}{4}$; see Figure~\ref{square-fig1}(c). In all cases we get a matching of size $m+1$ in $S$. 

\item {\em $\SM{3}{4}$ is the largest square.}
Consider the lines $l_1$ and $l_2$ which contain the top side and left side of $\SM{3}{4}$, respectively. If any point of $S_4$ is above $l_1$, then by induction, we get a matching of size $(m_1+1)+(m_2+1)+m_3+m_4$ in $\SM{1}{1}\cup \SP{1}{2}\cup\SM{3}{3}\cup \SM{3}{4}$. Otherwise, by induction, we get a matching of size $(m_1+1)+m_2+(m_3+1)+ m_4$ in $\SM{1}{1}\cup \SM{3}{2}\cup\SP{1}{3}\cup \SM{3}{4}$. In all cases we get a matching of size $m+1$ in $S$. 
\end{itemize}

{\bf Case 2:} {\em Two elements in $R$ are equal to 0 and two elements are equal to 1.}

In this case, we have $m = m_1 + m_2 + m_3 + m_4$. Because of the symmetry, only two cases may arise: (i) $r_1=r_2=1$ and $r_3=r_4=0$, (ii) $r_1=r_4=1$ and $r_2=r_3=0$. 
\begin{itemize}
 \item {\em $r_1=r_2=1$ and $r_3=r_4=0$.}
Consider the squares $\SM{3}{1}$, $\SM{3}{2}$, $\SM{2}{3}$, and $\SM{2}{4}$. By induction, we get matchings of size $m_1$, $m_2$, $m_3$, and $m_4$, in $\SM{3}{1}$, $\SM{3}{2}$, $\SM{2}{3}$, and $\SM{2}{4}$, respectively. Now consider the largest square among $\SM{3}{1}$, $\SM{3}{2}$, $\SM{2}{3}$, and $\SM{2}{4}$. Because of the symmetry, we have only two cases: (a) $\SM{3}{1}$ is the largest, (b) $\SM{2}{3}$ is the largest. In case (a) we get one more matched pair either in $\SP{1}{2}$ or in $\SP{2}{3}$. In case (b) we get one more matched pair either in $\SP{1}{1}$ or in $\SP{2}{4}$.

 \item {\em $r_1=r_4=1$ and $r_2=r_3=0$.}
Consider the squares $\SM{3}{1}$, $\SM{2}{2}$, $\SM{2}{3}$, and $\SM{3}{4}$. By induction, we get matchings of size $m_1$, $m_2$, $m_3$, and $m_4$, in $\SM{3}{1}$, $\SM{2}{2}$, $\SM{2}{3}$, and $\SM{3}{4}$, respectively. Now consider the largest square among $\SM{3}{1}$, $\SM{2}{2}$, $\SM{2}{3}$, and $\SM{3}{4}$. Because of the symmetry, we have only two cases: (a) $\SM{3}{1}$ is the largest, (b) $\SM{2}{2}$ is the largest. In case (a) we get one more matched pair either in $\SP{2}{2}$ or in $\SP{2}{3}$. In case (b) we get one more matched pair either in $\SP{1}{1}$ or in $\SP{1}{4}$.
\end{itemize}
\end{proof}

\section{A Conjecture on Strong Matching in $\G{\ddiscs}{P}$}
\label{conjecture-section}
In this section, we discuss a possible way to further improve upon Theorem~\ref{Gabriel-thr}, as well as
a construction leading to the conjecture that Algorithm~\ref{alg1} computes a strong matching of size at least $\lceil\frac{n-1}{8}\rceil$; unfortunately we are not able to prove this. 

In Section~\ref{Gabriel-section} we proved that $\mathcal{I}(e^+)$ contains at most 16 edges. In order to achieve this upper bound we used the fact that the centers of the disks in $\mathcal{I}(e^+)$ should be far apart. We did not consider the endpoints of the edges representing these disks. By Observation~\ref{no-point-in-circle-obs}, the disks representing the edges in $\mathcal{I}(e^+)$ cannot contain any of the endpoints. We applied this observation only on $u$ and $v$. Unfortunately, our attempts to apply this observation on the endpoints of edges in $\mathcal{I}(e^+)$ have been so far unsuccessful.

Recall that $T$ is a Euclidean minimum spanning tree of $P$, and for every edge $e=(u,v)$ in $T$, $\dg{e}$ is the degree of $e$ in $T(e^+)$, where $T(e^+)$ is the set of all edges of $T$ with weight at least $w(e)$. Note that $w(e)$ is directly related to the Euclidean distance between $u$ and $v$. Observe that the discs representing the edges adjacent to $e$ intersect $D(u,v)$. Thus, these edges are in $\Inf{e}$. 
We call an edge $e$ in $T$ a {\em minimal edge} if $e$ is not longer than any of its adjacent edges. We observed that the maximum degree of a minimal edge is an upper bound for $\Inf{e}$. We conjecture that,

\begin{conjecture}
{\em \Inf{$T$}} is at most the maximum degree of a minimal edge.
\end{conjecture}

Monma and Suri~\cite{Monma1992} showed that for every point set $P$ there exists a Euclidean minimum spanning tree, $MST(P)$, of maximum vertex degree five. Thus, the maximum edge degree in $MST(P)$ is 9. We show that for every point set $P$, there exists a Euclidean minimum spanning tree, $MST(P)$, such that the degree of each node is at most five and the degree of each minimal edge is at most eight. This implies the conjecture that $\Inf{MST(P)}\le 8$. That is, Algorithm~\ref{alg1} returns a strong matching of size at least $\lceil\frac{n-1}{8}\rceil$.

\begin{figure}[ht]
  \centering
    \includegraphics[width=0.4\textwidth]{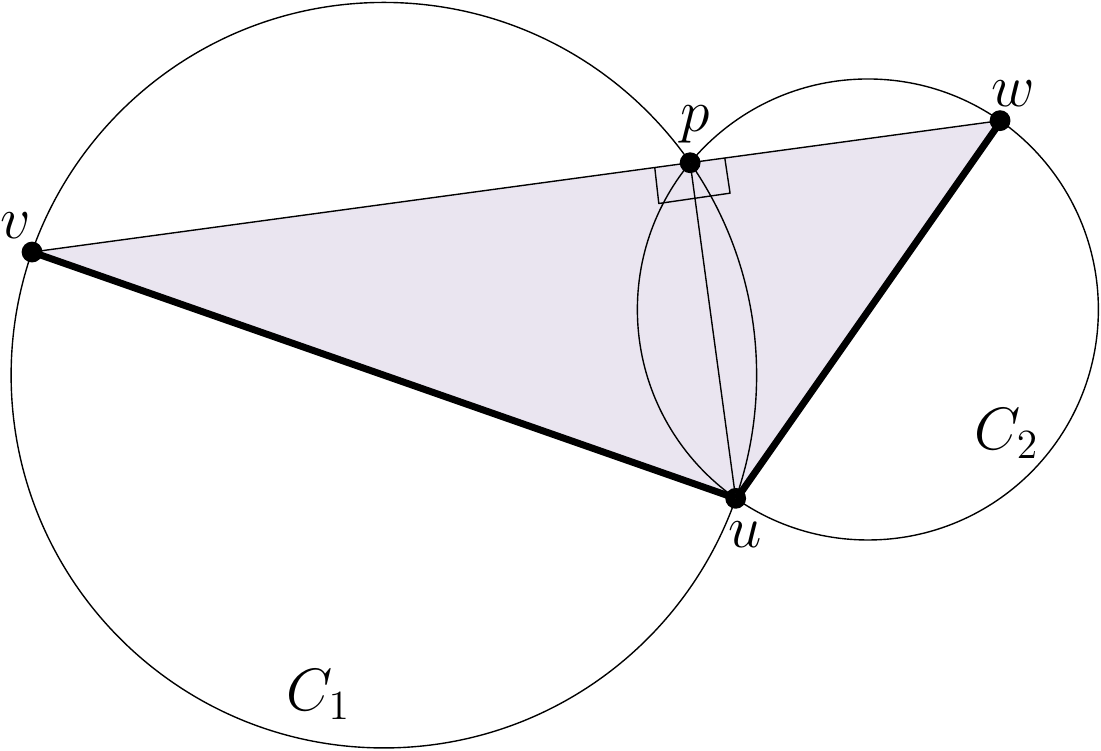}
  \caption{In $MST(P)$, the triangle $\bigtriangleup uvw$ formed by two adjacent edges $uv$ and $uw$, is empty.}
\label{empty-triangle-figure}
\end{figure}
\begin{lemma}
\label{empty-triangle-lemma}
If $uv$ and $uw$ are two adjacent edges in $MST(P)$, then the triangle $\bigtriangleup uvw$ has no point of $P\setminus\{u, v, w\}$ in its interior or on its boundary.
\end{lemma}
\begin{proof}
If the angle between $uv$ and $uw$ is equal to $\pi$, then there is no other point of $P$ on $uv$ and $uw$. Assume that $\angle vuw < \pi$. Refer to Figure \ref{empty-triangle-figure}. Since $MST(P)$ is a subgraph of the Gabriel graph, the circles $C_1$ and $C_2$ with diameters $uv$ and $uw$ are empty. Since $\angle vuw < \pi$, $C_1$ and $C_2$ intersect each other at two points, say $u$ and $p$. Connect $u$, $v$ and $w$ to $p$. Since $uv$ and $uw$ are the diameters of $C_1$ and $C_2$, $\angle upv=\angle wpu=\pi/2$.
This means that $vw$ is a straight line segment. Since $C_1$ and $C_2$ are empty and $\bigtriangleup uvw \subset C_1 \cup C_2$, it follows that $\bigtriangleup uvw \cap P = \{u, v, w\}$.
\end{proof}

\begin{figure}[htb]
  \centering
  \includegraphics[width=.4\textwidth]{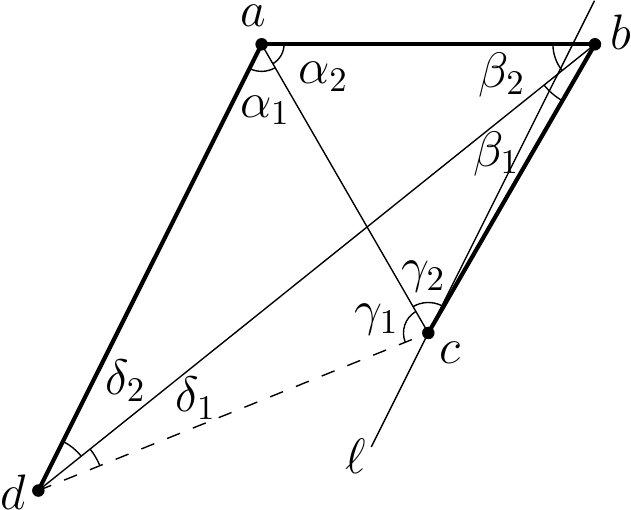}
  \caption{Illustration of Lemma~\ref{convex-quadrilateral-lemma}: $|ab|\le|bc|\le|ad|$, $\angle abc\ge \pi/3$, $\angle bad\ge \pi/3$, and $\angle abc + \angle bad \le \pi$.}
\label{convex-quadrilateral-fig}
\end{figure}

\begin{lemma}
\label{convex-quadrilateral-lemma}
Follow Figure~\ref{convex-quadrilateral-fig}. For a convex-quadrilateral $Q=a,b,c,d$ with $|ab|\le|bc|\le|ad|$, if  $\min\{\angle abc, \angle bad\}\ge \pi/3$ and $\angle abc + \angle bad \le \pi$, then $|cd|\le |ad|$.
\end{lemma}
\begin{proof}
 Let $\alpha_1=\angle cad$, $\alpha_2=\angle bac$, $\beta_1=\angle cbd$, $\beta_2=\angle abd$, $\gamma_1=\angle acd$, $\gamma_2=\angle acb$, $\delta_1=\angle bdc$, and $\delta_2=\angle adb$; see Figure~\ref{convex-quadrilateral-fig}. Since $|ab|\le|bc|\le |ad|$, $$\gamma_2\le \alpha_2\text{ and }\delta_2\le \beta_2.$$ Let $\ell$ be a line passing through $c$ which is parallel to $ad$. Since $\angle abc + \angle bad \le \pi$, $\ell$ intersects the line segment $ab$. This implies that $\alpha_1\le \gamma_2$. If $\beta_1<\delta_1$, then $|cd|<|bc|$, and hence $|cd|<|ad|$ and we are done. Assume that $\delta_1\le \beta_1$. In this case, $\delta\le \beta$. Now consider the two triangles $\bigtriangleup abc$ and $\bigtriangleup acd$. Since $\delta\le \beta$ and $\alpha_1\le\gamma_2$, $\alpha_2\le \gamma_1$. Then we have $$\alpha_1\le \gamma_2\le \alpha_2\le\gamma_1.$$

Since $\alpha_1\le\gamma_1$, $|cd|\le |ad|$, where the equality holds only if $\alpha_1= \gamma_2= \alpha_2=\gamma_1$, i.e., $Q$ is a diamond. This completes the proof.
\end{proof}

\begin{figure}[htb]
  \centering
  \includegraphics[width=.7\textwidth]{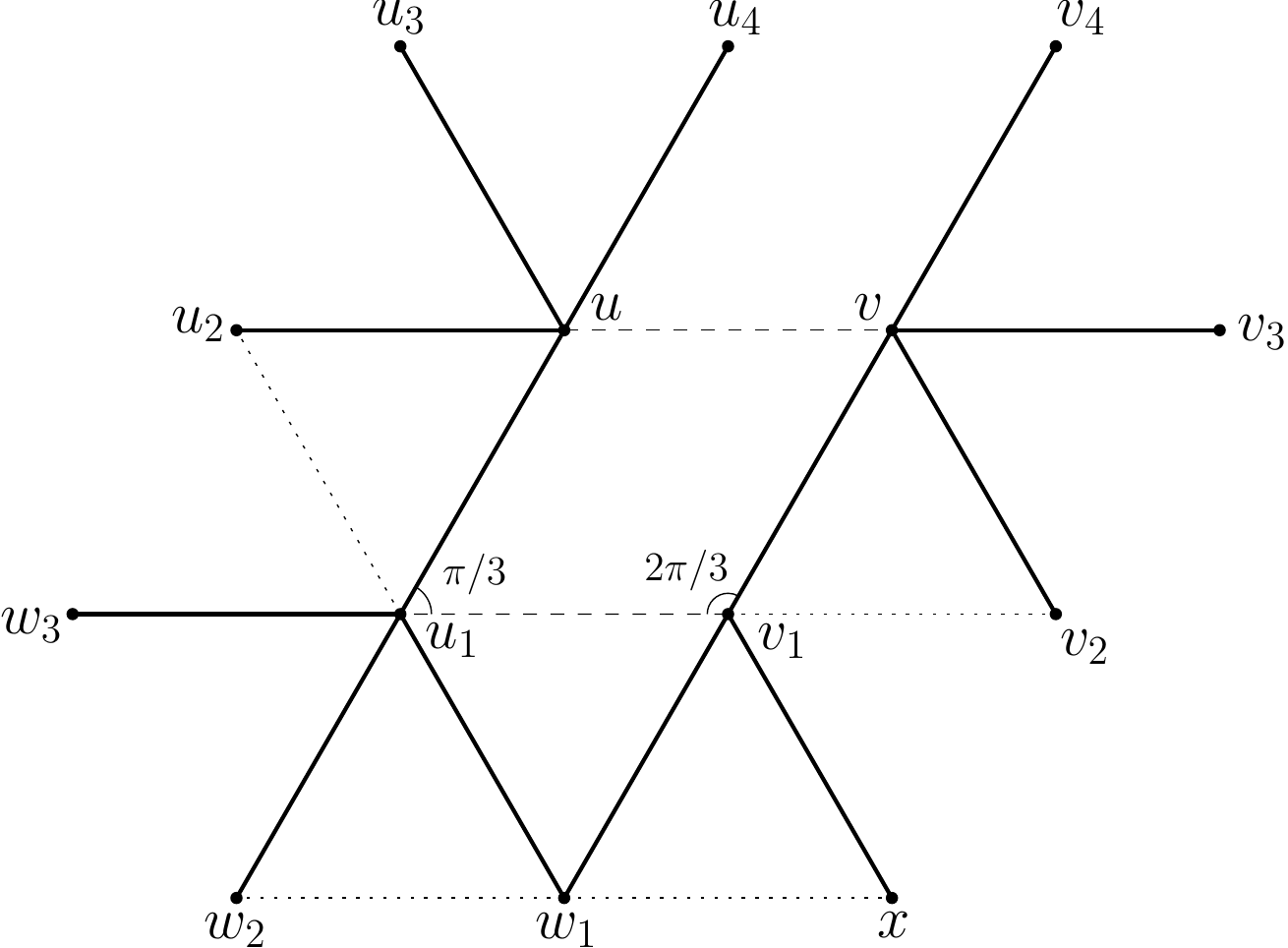}
  \caption{Solid segments represent the edges of $MST(P)$. Dashed segments represent the swapped edges. Dotted segments represent the edges which cannot exist.}
\label{degree8-fig}
\end{figure}

\begin{lemma} 
\label{edge-degree-lemma}
Every finite set of points $P$ in the plane admits a minimum spanning tree whose node degree is at most five and whose minimal-edge degree is at most nine.
\end{lemma}
\begin{proof}
Consider a minimum spanning tree, $MST(P)$, of maximum vertex degree 5. The maximum edge degree in $MST(P)$ is 9. Consider any minimal edge, $uv$. If the degree of $uv$ is 8, then $MST(P)$ satisfies the statement of the lemma. Assume that the degree of $uv$ is 9. Let $u_1, u_2,u_3,u_4$ and $v_1, v_2,v_3,v_4$ be the the neighbors of $u$ and $v$ in clockwise and counterclockwise orders, respectively. See Figure~\ref{degree8-fig}. In $MST(P)$, the angles between two adjacent edges are at least $\pi/3$. Since $\angle u_iuu_{i+1}\ge \pi/3$ and $\angle v_ivv_{i+1}\ge \pi/3$ for $i=1,2,3$, either $\angle vuu_1+\angle uvv_1\le \pi$ or $\angle vuu_4+\angle uvv_4\le \pi$. Without loss of generality assume that $\angle vuu_1+\angle uvv_1\le \angle vuu_4+\angle uvv_4$ and $\angle vuu_1+\angle uvv_1\le \pi$. We prove that the spanning tree
obtained by swapping the edge $uv$ with $u_1v_1$ is also a minimum spanning tree, and it has one fewer minimal-edge of degree 9. By repeating this procedure at each minimal-edge of degree 9, we obtain a minimum spanning tree which satisfies the statement of the lemma. 
Let $Q=u,v,v_1,u_1$. By Lemma~\ref{empty-triangle-lemma}, $v_1$ is outside the triangle $\bigtriangleup u_1uv$, and $u_1$ is outside the triangle $\bigtriangleup uvv_1$. In addition, $u_1$ and $v_1$ are on the same side of the line subtended from $uv$. Thus, $Q$ is a convex quadrilateral.
Without loss of generality assume that $|vv_1|\le |uu_1|$. By Lemma~\ref{convex-quadrilateral-lemma}, $|u_1v_1|\le |uu_1|$. If $|u_1v_1|< |uu_1|$, we get a contradiction to Lemma~\ref{cycle-lemma}. Thus, assume that $|u_1v_1|= |uu_1|$. As shown in the proof of Lemma~\ref{convex-quadrilateral-lemma}, this case happens only when $Q$ is a diamond. This implies that $\angle vuu_1+\angle uvv_1=\pi$, and consequently $\angle vuu_4+\angle uvv_4= \pi$. In addition, $\angle u_iuu_{i+1}= \pi/3$ and $\angle v_ivv_{i+1}= \pi/3$ for $i=1,2,3$. To establish the validity of our edge-swap, observe that the nine edges incident to $u$ and $v$ are all equal in length. Therefore, swapping $uv$ with $u_1v_1$ does not change the cost of the spanning tree and, furthermore, the resulting tree is a valid spanning tree
since $u_1v_1$ is not an edge of the original spanning tree $MST(P)$; otherwise $u,v,v_1$, and $u_1$ would form a cycle. We have removed a minimal edge $uv$ of degree 9, but it remains to show that the degree of $u_1$ and $v_1$ does not increase to six and new minimal edge of degree 9 is not generated. Note that $u_1u_2$ and $v_1v_2$ are not the edges of $MST(P)$, and hence, $\dg{u_1}$ and $\dg{v_1}$ are still less than six. In order to show that no new minimal edge is generated, we differentiate between two cases:

\begin{itemize}
 \item $\min\{\angle vv_1u_1, \angle v_1u_1u\} > \pi/3$. Since $\angle v_1u_1u>\pi/3$ and $\angle uu_1u_2= \pi/3$, $u_1$ can be adjacent to at most two vertices other than $u$ and $v_1$, and hence $\dg{u_1}\le 4$; similarly $\dg{v_1}\le 4$. Thus, $u$, $v$, $u_1$, and $v_1$ are of degree at most four, and hence no new minimal edge of degree 9 is generated.

  \item $\min\{\angle vv_1u_1, \angle v_1u_1u\} = \pi/3$. W.l.o.g. assume that $\angle v_1u_1u=\pi/3$. This implies that $\angle vv_1u_1= 2\pi/3$. Since $\angle v_1u_1u=\pi/3$ and $\angle uu_1u_2= \pi/3$, $u_1$ is adjacent to at most three vertices other than $u$ and $v_1$. Let $u,v_1, w_1,w_2,w_3$ be the neighbors of $u_1$ in clockwise order. Note that $v_1$ is not adjacent to $u$, $v_2$ nor $w_1$. But $v_1$ can be connected to another vertex, say $x$, which implies that $\dg{v_1}\le 3$. We prove that the spanning tree obtained by swapping the edge $u_1v_1$ with $v_1w_1$ is also a minimum spanning tree of node degree at most five, which has one fewer minimal edge of degree 9. The new tree is a legal minimum spanning tree for $P$, because $|v_1w_1|=|v_1u_1|$. In addition, $\dg{u_1}\le 4$ and $\dg{v_1}\le 4$. Since $w_1w_2$ and $w_1x$ are illegal edges, $\dg{w_1}\le 4$. Thus, $u$, $v$, $u_1$, $v_1$, and $w_1$ are of degree at most four and no new minimal edge of degree 9 is generated. This completes the proof that our edge-swap reduces the number of minimal-edges of degree nine by one.
\end{itemize}
\end{proof}

\section{Conclusion}
\label{conclusion}

Given a set of $n$ points in general position in the plane, we considered the problem of strong matching of points with convex geometric shapes. A matching is strong if the objects representing whose edges are pairwise disjoint. In this paper we presented algorithms which compute strong matchings of points with diametral-disks, equilateral-triangles, and squares. Specifically we showed that:
\begin{itemize}
 \item There exists a strong matching of points with diametral-disks of size at least $\lceil\frac{n-1}{17}\rceil$.
 \item There exists a strong matching of points with downward equilateral-triangles of size at least $\lceil\frac{n-1}{9}\rceil$.
\item There exists a strong matching of points with downward/upward equilateral-triangles of size at least $\lceil\frac{n-1}{4}\rceil$.
\item There exists a strong matching of points with axis-parallel squares of size at least $\lceil\frac{n-1}{4}\rceil$.
\end{itemize}
 
The existence of a downward/upward equilateral-triangle matching of size at least $\lceil\frac{n-1}{4}\rceil$, implies the existence of either a downward equilateral-triangle matching of size at least $\lceil\frac{n-1}{8}\rceil$ or an upward equilateral-triangle matching of size at least $\lceil\frac{n-1}{8}\rceil$. This does not imply a lower bound better than $\lceil\frac{n-1}{9}\rceil$ for downward equilateral-triangle matching (or any fixed oriented equilateral-triangle).

A natural open problem is to improve any of the provided lower bounds, or extend these results for other convex shapes. The specific open problem is to prove that Algorithm~\ref{alg1} computes a strong matching of points with diametral-disks of size at least $\lceil\frac{n-1}{8}\rceil$ as discussed in Section~\ref{conjecture-section}.
\bibliographystyle{abbrv}
\bibliography{Strong-Matching.bib}
\end{document}